\documentclass[11pt]{article}

\usepackage{algorithm}
\usepackage{amssymb,amsmath,amscd,latexsym,epic,eepic}
\usepackage{luca} 
\usepackage{macro} 
\usepackage{times}
\usepackage{fullpage}
\usepackage{gastex}

\newcommand{\wh}{\widehat}

\begin{document}

\sloppy

\title{Strategy Improvement for \\
Concurrent Reachability and Safety 
Games
\thanks{Corresponding Author: Krishnendu Chatterjee, 
email: krish.chat@ist.ac.at, 
Address: Am Campus 1, IST Austria, Klosterneuburg, A3400, Austria.
Telephone number: +43-2243-9000-3201, Fax Number: +43-2243-9000-2000.}
$\ ^,$\thanks{This paper is 
an improved version of the combined results that appeared 
in~\cite{CdAH06,CdAH09}: this paper is a joint paper that 
combines the results of~\cite{CdAH06,CdAH09}, and presents detailed 
proofs of all the results.}
$\ ^,$\thanks{{\bf There is a serious and irreparable error in Theorem~4.3 
of~\cite{CdAH09} regarding the convergence property of the improvement algorithm for safety games.} 
This is illustrated in Example~\ref{ex:counter-soda}.
In the present version we prove all the required properties for a modified 
algorithm (Theorem~8).
We thank anonymous reviewers for many insightful comments that helped us immensely, 
and warmly acknowledge their help.} 
}

\author{Krishnendu Chatterjee$^\dag$ \qquad
  Luca de Alfaro$^{\S}$ \qquad
  Thomas A. Henzinger$^{\dag\ddag}$\\[5pt]
\normalsize
  $\strut^\dag$ IST Austria (Institute of Science and Technology Austria) \\
\normalsize  $\strut^\S$ CE, University of California, Santa Cruz,USA\\
\normalsize
  $\strut^\ddag$ Computer and Communication Sciences, EPFL, Switzerland\\
\normalsize
  \texttt{$\{$krish.chat,tah$\}$@ist.ac.at, luca@soe.ucsc.edu}
}

\date{ 
}

\maketitle

\begin{abstract}
We consider concurrent games played on graphs.  
At every round of a game, each player simultaneously and independently 
selects a move;
the moves jointly determine the transition to a successor state. 
Two basic objectives are the safety objective to stay forever in a given 
set of states, and its dual, the reachability objective to reach a given 
set of states.
First, we present a simple proof of the fact that in concurrent reachability 
games, for all $\vare>0$, memoryless $\vare$-optimal strategies exist.  
A memoryless strategy is independent of the history of plays, and an 
$\vare$-optimal strategy achieves the objective with probability within 
$\vare$ of the value of the game.  
In contrast to previous proofs of this fact, 
our proof is more elementary and more combinatorial.  
Second, we present a strategy-improvement (a.k.a.\ policy-iteration) algorithm 
for concurrent games with reachability objectives. 
We then present a strategy-improvement algorithm for concurrent games with 
safety objectives. 
Our algorithms yield sequences of player-1 strategies which ensure
probabilities of winning that converge monotonically to the value of the game. 
Our result is significant because the strategy-improvement algorithm 
for safety games provides, for the first time, a way to approximate the 
value of a concurrent safety game {\em from below}.
Previous methods could approximate the values of these games only from one 
direction, and as no rates of convergence are known, they did not provide a
practical way to solve these games.
\end{abstract}


\section{Introduction}

We consider games played between two players on graphs. 
At every round of the game, each of the two players selects a move; 
the moves of the players then determine the transition to the successor
state. 
A play of the game gives rise to a path in the graph. 
We consider the two basic objectives for the players: 
{\em reachability} and {\em safety}. 
The reachability goal asks player~1 to reach a given set of target states 
or, if randomization is needed to play the game, to maximize the probability 
of reaching the target set. 
The safety goal asks player~2 to ensure that a given set of safe states
is never left or, if randomization is required, to minimize the probability 
of leaving the target set.
The two objectives are dual, and the games are determined: 
the supremum probability with which player~1 can reach the target set is equal 
to one minus the supremum probability with which player~2 can confine the game 
to the complement of the target set~\cite{Eve57}.

These games on graphs can be divided into two classes: 
{\em turn-based\/} and {\em concurrent}.
In turn-based games, only one player has a choice of moves at each
state; 
in concurrent games, at each state both players choose a move,
simultaneously and independently, from a set of available moves. 
For turn-based games, the solution of games with reachability and safety 
objectives has long been known. 
If each move determines a unique successor state, then the games are 
P-complete and can be solved in linear time in the size of the game graph.
If, more generally, each move determines a probability distribution on 
possible successor states, then the problem of deciding whether a 
turn-based game can 
be won with probability greater than a given threshold $p\in [0,1]$ is in 
NP $\cap$ co-NP \cite{Con92}, and the exact value of the game can be 
computed by a strategy-improvement algorithm~\cite{Con93}, which works well 
in practice. 
These results all depend on the fact that in turn-based reachability and
safety games, both players have optimal deterministic 
(i.e., no randomization is required), memoryless strategies.  
These strategies are functions from states to moves, so they are finite in 
number, and this guarantees the termination of the strategy-improvement 
algorithm. 

The situation is very different for concurrent games.
The player-1 {\em value\/} of the game is defined, as usual, as the sup-inf 
value: 
the supremum, over all strategies of player~1, of the infimum, over all 
strategies of player~2, of the probability of achieving the reachability 
or safety goal. 
In concurrent reachability games, player~1 is guaranteed only the existence 
of $\varepsilon$-optimal strategies, which ensure that the value of the game 
is achieved within a specified tolerance $\varepsilon>0$ \cite{Eve57}.
Moreover, while these strategies (which depend on $\varepsilon$) are 
memoryless, in general they require randomization~\cite{Eve57}  
(even in the special case in which the transition function is 
deterministic).
For player~2 (the safety player), {\em optimal\/} memoryless strategies 
exist~\cite{TParth71}, which again require randomization (even when the 
transition function is deterministic).
All of these strategies are functions from states to probability 
distributions on moves. 
The question of deciding whether a concurrent game can be won with 
probability greater than $p$ is in PSPACE;
this is shown by reduction to the theory of the real-closed fields 
\cite{EY06}. 

To summarize:
while strategy-improvement algorithms are available for turn-based 
reachability and safety games~\cite{Con93}, 
so far no strategy-improvement algorithms or even approximation schemes 
were known for concurrent games.
If one wanted to compute the value of a concurrent game within a specified 
tolerance $\varepsilon>0$, one was reduced to using a binary search 
algorithm that approximates the value by iterating queries in the theory of 
the real-closed fields.
Value-iteration schemes were known for such games, but they can be used to 
approximate the value from one direction only, for reachability goals from 
below, and for safety goals from above~\cite{dAM04}.
The value-iteration schemes are not guaranteed to terminate.
Worse, since no convergence rates are known for these schemes, they 
provide no termination criteria for approximating a value 
within~$\varepsilon$.

\medskip\noindent{\bf Our results for concurrent reachability games.}
Concurrent reachability games belong to the family of stochastic games
\cite{Sha53,Eve57}, and they have been studied more
specifically in \cite{crg-tcs07,dAH00,dAM04}. 
Our contributions for concurrent reachability games are two-fold.
First, we present a simple and combinatorial proof of the
existence of memoryless $\vare$-optimal strategies for concurrent 
games with reachability objectives, for all $\vare>0$.
Second, using the proof techniques we developed for proving existence 
of memoryless $\vare$-optimal strategies, for $\vare>0$, we obtain a 
strategy-improvement (a.k.a.\ policy-iteration) algorithm for concurrent 
reachability games.
Unlike in the special case of turn-based games 
the algorithm need not terminate in finitely many iterations.

It has long been known that optimal strategies need not exist for
concurrent reachability games, 
and for all $\vare > 0$, there exist $\vare$-optimal
strategies that are memoryless~\cite{Eve57}. 
A proof of this fact can be obtained by considering limit of discounted games.
The proof considers \emph{discounted} versions of reachability games,
where a play that reaches the target in $k$ steps is assigned a value of
$\alpha^k$, for some discount factor $0 < \alpha \leq 1$. 
It is possible to show that, for $0 < \alpha < 1$, 
memoryless optimal strategies  exist. 
The result for the undiscounted ($\alpha = 1$) case followed from an 
analysis of the limit behavior of such optimal strategies for
$\alpha \rightarrow 1$.
The limit behavior is studied with the help of
results from the field of real Puisieux series~\cite{MN81}. 
This proof idea works not only for reachability games, but also for
total-reward games with nonnegative rewards (see~\cite[Chapter~5]{FV97} 
for details).
A more recent result \cite{EY06} establishes the existence of 
memoryless $\vare$-optimal strategies for certain infinite-state 
(recursive) concurrent games, but again the proof relies on results 
from analysis and  properties of solutions of certain polynomial 
functions. 
Another proof of existence of memoryless $\vare$-optimal strategies for 
reachability objectives follows from the result of~\cite{Eve57} and 
the proof uses induction on the number of states of the game.
We show  the existence of memoryless $\vare$-optimal strategies
for concurrent reachability games by more combinatorial and elementary means.
Our proof relies only on combinatorial techniques and on simple 
properties of Markov decision processes \cite{Bertsekas95,luca-thesis}. 
As our proof is more combinatorial, we believe that the proof
techniques  will find future applications in game theory. 

Our proof of the existence of memoryless $\vare$-optimal strategies, for
all $\vare > 0$, is built upon a value-iteration scheme that
converges to the value of the game \cite{dAM04}. 
The value-iteration scheme computes a sequence $u_0, u_1, u_2, \ldots$
of valuations, where for $i = 0, 1, 2, \ldots$ each valuation $u_i$
associates with each state $s$ of the game a lower bound $u_i(s)$ on
the value of the game, such that $\lim_{i \go \infty}
u_i(s)$ converges to the value of the game at $s$. 
The convergence is monotonic from below, but no rate of convergence was known.
From each valuation $u_i$, we can extract a memoryless,
randomized player-1 strategy, by considering the (randomized) choice
of moves for player~1 that achieves the maximal one-step expectation
of $u_i$. 
In general, a strategy $\stra_i$ obtained in this fashion is not
guaranteed to achieve the value $u_i$.
We show that $\stra_i$ is guaranteed to achieve the value $u_i$ if it
is \emph{proper}, that is, if regardless of the strategy adopted by
player~2, the play reaches with probability~1 states that are either
in the target, or that have no path leading to the target. 
Next, we show how to extract from the sequence of valuations $u_0,
u_1, u_2, \ldots$ a sequence of memoryless randomized player-1
strategies $\stra_0, \stra_1, \stra_2, \ldots$ that are guaranteed to
be proper, and thus achieve the values $u_0, u_1, u_2, \ldots$. 
This proves the existence of memoryless $\vare$-optimal strategies for
all $\vare > 0$. 
Our proof is completely different as compared to the proof of~\cite{Eve57}: 
the proof of~\cite{Eve57} uses induction on the number of states, whereas our 
proof is based on the notion of ranking function obtained from the 
value-iteration algorithm.

We then apply the techniques developed for the above proof to design 
a \emph{strategy-improvement} algorithm for concurrent reachability 
games. 
Strategy-improvement algorithms, also known as \emph{policy-iteration} 
algorithms in the context of Markov decision processes~\cite{Howard}, 
compute a sequence of memoryless strategies 
$\stra'_0, \stra'_1, \stra'_2, \ldots$ such that, for all $k \geq 0$,
(i)~the strategy $\stra'_{k+1}$ is at all states no worse than 
$\stra'_{k}$; 
(ii)~if $\stra'_{k+1} = \stra'_k$, then $\stra_k$ is optimal; and 
(iii)~for every $\vare > 0$, we can find a $k$ sufficiently large so that
$\stra'_k$ is $\vare$-optimal. 
Computing a sequence of strategies $\stra_0, \stra_1, \stra_2, \ldots$
on the basis the value-iteration scheme from above does not yield a
strategy-improvement algorithm, as condition (ii) may be violated:
there is no guarantee that a step in the value iteration leads to an
improvement in the strategy. 
We will show that the key to obtain a strategy-improvement algorithm
consists in recomputing, at each iteration, the values of the player-1
strategy to be improved, and in adopting a particular strategy-update
rule, which ensures that all generated strategies are proper. 
Unlike previous proofs of strategy-improvement algorithms for concurrent
games \cite{Con93,FV97}, which rely on the analysis of
discounted versions of the games, our analysis is again more combinatorial. 
Hoffman-Karp~\cite{HofKar66} presented a strategy improvement algorithm 
for the special case of concurrent games with ergodic property (i.e., 
from every state $s$ any other state $t$ can be guaranteed to reach 
with probability~1) (also see algorithm for discounted games in~\cite{RCN73}). 
Observe that for concurrent reachability games, with the ergodic assumption
the value at all states is trivially~1, and thus the ergodic assumption 
gives us the trivial case. 
Our results give a combinatorial strategy improvement algorithm for the 
whole class of concurrent reachability games.
The results of~\cite{EY06} presents a strategy improvement algorithm for 
recursive concurrent games with termination criteria: the algorithm 
of~\cite{EY06} is more involved (depends on properties of certain polynomial 
functions) and works for the more general class of recursive concurrent games.
Differently from turn-based games \cite{Con93}, for concurrent
games we cannot guarantee the termination of the strategy-improvement
algorithm. 
However, for turn-based stochastic games we present a detailed analysis of 
termination criteria. 
Our analysis is based on bounds on the precision of values for turn-based 
stochastic games.
As a consequence of our analysis, we obtain an improved upper bound for 
termination for turn-based stochastic games.

\medskip\noindent{\bf Our results for concurrent safety games.}
We present for the first time a strategy-improvement scheme 
that approximates the value of a concurrent safety game {\em from below}.
Together with the strategy improvement algorithm for reachability games, 
or the value-iteration scheme, to approximate the value of such a game 
from above, we obtain a termination criterion for computing the value 
of concurrent reachability and safety games within any given tolerance 
$\varepsilon>0$.
This is the first termination criterion for an algorithm that approximates 
the value of a concurrent game.
Several difficulties had to be overcome in developing our scheme.
First,
while the strategy-improvement algorithm that approximates reachability 
values from below is based on locally improving a strategy 
on the basis of the valuation it yields, this approach does not suffice 
for approximating safety values from below:
we would obtain an increasing sequence of values, but they would not 
necessarily converge to the value of the game 
(see Example~\ref{examp:conc-safety}). 
Rather, we introduce a novel, non-local improvement step, which augments 
the standard valuation-based improvement step.
Each non-local step involves the solution of an appropriately constructed 
turn-based game.  
The turn-based game constructed is polynomial in the state space of the original 
game, but \emph{exponential} in the number of actions. 
It is an interesting open question whether the turn-based game can be 
also made polynomial in the number of the actions.
Second,
as value-iteration for safety objectives converges from above, while our
sequences of strategies yield values that converge from below, the
proof of convergence for our algorithm cannot be derived from a
connection with value-iteration, as was the case for reachability
objectives. 
We had to develop new proof techniques both to show the monotonicity
of the strategy values produced by our algorithm, and to show their
convergence to the value of the game.

\medskip\noindent{\bf Added value of our algorithms.} The new strategy 
improvement algorithms we present in this paper has two important contributions
as compared to the classical value-iteration algorithms.

\begin{enumerate}

\item \emph{Termination for approximation.} 
The value-iteration algorithm for reachability games converges from below, 
and the value-iteration for safety games converges for above. 
Hence given desired precision $\vare>0$ for approximation, there is no 
termination criteria to stop the value-iteration algorithm and guarantee 
$\vare$-approximation.
The sequence of valuation of our strategy improvement algorithm for concurrent 
safety games converges from below, and along with the value-iteration or 
strategy improvement algorithm for concurrent reachability games we obtain the 
\emph{first} termination criteria for $\vare$-approximation of values in 
concurrent reachability and safety games.
Using a result of~\cite{HKM09} on the bound on $k$-uniform memoryless 
$\vare$-optimal strategies, for $\vare>0$, we also obtain a bound on the number 
of iterations of the strategy improvement algorithms that guarantee 
$\vare$-approximation of the values. 
Moreover a recent result of~\cite{CSR11} provide a nearly tight 
double exponential upper and lower bound on the number of iterations required
for $\vare$-approximation of the values.

\item \emph{Approximation of strategies.}
Our strategy improvement algorithms are also the first approach 
to approximate memoryless $\vare$-optimal strategies in concurrent reachability and 
safety games.
The witness strategy produced by the value-iteration algorithm for concurrent 
reachability games is not memoryless; and for concurrent safety games since
the value-iteration algorithm converges from above it does not provide any 
witness strategies. 
Our strategy improvement algorithms for concurrent reachability and safety 
games yield sequence of memoryless strategies that ensure for convergence 
to the value of the game from below, and yield witness memoryless strategies 
to approximate the value of concurrent reachability and safety games.

\end{enumerate}

\section{Definitions} 

\noindent{\bf Notation.} 
For a countable set~$A$, a {\em probability distribution\/} on $A$ is a 
function $\trans\!:A\to[0,1]$ such that $\sum_{a \in A} \trans(a) = 1$. 
We denote the set of probability distributions on $A$ by $\distr(A)$. 
Given a distribution $\trans \in \distr(A)$, we denote by 
$\supp(\trans) = \{x \in A \mid \trans(x) > 0\}$ the support set
of $\trans$.

\begin{definition}{(Concurrent games).}
A (two-player) {\em concurrent game structure\/} $\gamegraph = \langle S,
\moves, \mov_1, \mov_2, \trans \rangle$ consists of the following
components: 
\begin{itemize}

\item A finite state space $S$ and a finite set $\moves$ of moves or actions.

\item Two move assignments $\mov_1, \mov_2 \!: S\to 2^\moves
	\setminus \emptyset$.  For $i \in \{1,2\}$, assignment
	$\mov_i$ associates with each state $s \in S$ a nonempty
	set $\mov_i(s) \subseteq \moves$ of moves available to player $i$
	at state $s$.  

\item 
A probabilistic transition function 
$\trans: S \times \moves \times \moves \to \Distr(S)$ that gives the
probability $\trans(s, a_1, a_2)(t)$ of a transition from $s$ to
$t$ when player~1 chooses at state $s$ move $a_1$ and player~2 chooses 
move $a_2$, for all $s,t\in S$ and $a_1 \in \mov_1(s)$, $a_2 \in \mov_2(s)$.  
\end{itemize}
\end{definition}

\noindent
We denote by $|\trans|$ the size of transition 
function, i.e., $|\trans|=\sum_{s\in S,a \in \mov_1(s),b\in \mov_2(s),t\in S} 
|\trans(s,a,b)(t)|$, where $|\trans(s,a,b)(t)|$ is the number of bits 
required to specify the transition probability $\trans(s,a,b)(t)$.
We denote by $|\gamegraph|$ the size of the game graph, and $|G|=|\trans|+|S|$.
At every state $s\in S$, player~1 chooses a move $a_1\in\mov_1(s)$,
and simultaneously and independently player~2 chooses a move $a_2\in\mov_2(s)$.
The game then proceeds to the successor state $t$ with probability
$\trans(s,a_1,a_2)(t)$, for all $t \in S$.
A state $s$ is an \emph{absorbing state} if for all 
$a_1 \in \mov_1(s)$ and $a_2 \in \mov_2(s)$, we have 
$\trans(s, a_1,a_2)(s)=1$.
In other words, at an absorbing state $s$ for all choices of moves of the 
two players, the successor state is always $s$. 

\begin{definition}{(Turn-based stochastic games).}
A \emph{turn-based stochastic game graph} 
(\emph{$2\half$-player game graph})
$\gamegraph =\langle (S, E), (S_1,S_2,S_R),\trans\rangle$ 
consists of a finite directed graph $(S,E)$, a partition $(S_1$, $S_2$,
$S_R)$ of the finite set $S$ of states, and a probabilistic transition 
function $\trans$: $S_R \rightarrow \distr(S)$, where $\distr(S)$ denotes the 
set of probability distributions over the state space~$S$. 
The states in $S_1$ are the {\em player-$1$\/} states, where player~$1$
decides the successor state; the states in $S_2$ are the {\em 
player-$2$\/} states, where player~$2$ decides the successor state; 
and the states in $S_R$ are the {\em random or probabilistic\/} states, where
the successor state is chosen according to the probabilistic transition
function~$\trans$. 
We assume that for $s \in S_R$ and $t \in S$, we have $(s,t) \in E$ 
iff $\trans(s)(t) > 0$, and we often write $\trans(s,t)$ for $\trans(s)(t)$. 
For technical convenience we assume that every state in the graph 
$(S,E)$ has at least one outgoing edge.
For a state $s\in S$, we write $E(s)$ to denote the set 
$\set{t \in S \mid (s,t) \in E}$ of possible successors.
We denote by $|\trans|$ the size of the transition 
function, i.e., $|\trans|=\sum_{s\in S_R,t\in S} 
|\trans(s)(t)|$, where $|\trans(s)(t)|$ is the number of bits 
required to specify the transition probability $\trans(s)(t)$.
We denote by $|\gamegraph|$ the size of the game graph, and 
$|G|=|\trans|+|S|+|E|$.
\end{definition}

\smallskip\noindent{\bf Plays.}
A \emph{play} $\pat$ of $\gamegraph$ is an infinite sequence
$\pat = \langle s_0,s_1,s_2,\ldots \rangle $ of states in $S$ such that 
for all $k\ge 0$, there are moves $a^k_1 \in \mov_1(s_k)$ and 
$a^k_2 \in \mov_2(s_k)$ with $\trans(s_k,a^k_1,a^k_2)(s_{k+1}) >0$.
We denote by $\Paths$ the set of all plays, and by $\Paths_s$ the set of all 
plays $\pat=\seqs$ such that $s_0=s$, that is, the set of plays starting 
from state~$s$. 

\medskip\noindent{\bf Selectors and strategies.}
A \emph{selector} $\xi$ for player $i \in \set{1,2}$ is a function
$\xi : S \to \Distr(\moves)$ such that for all states $s \in S$ and
moves $a \in \moves$, if $\xi(s)(a) > 0$, then $a \in \mov_i(s)$.
A selector $\xi$ for player $i$ at a state $s$ is a distribution 
over moves such that if $\xi(s)(a)>0$, then $a \in \mov_i(s)$.
We denote by $\Sel_i$ the set of all selectors for player~$i \in \set{1,2}$,
and similarly, we denote by $\Sel_i(s)$ the set of all selectors for
player~$i$ at a state $s$.
The selector $\xi$ is {\em pure\/} if for every state $s \in S$, there is 
a move $a \in \moves$ such that $\xi(s)(a) = 1$.
A \emph{strategy} for player $i\in\set{1,2}$ is a function 
$\stra: S^+ \to \Distr(\moves)$ that
associates with every finite, nonempty sequence of states,
representing the history of the play so far, a selector for player~$i$; 
that is, for all $w \in S^*$ and $s \in S$, we have 
$\supp(\stra(w \cdot s)) \subs \mov_i(s)$.
The strategy $\stra$ is {\em pure\/} 
if it always chooses a pure selector;
that is, for all $w \in S^+$, there is a move $a \in \moves$ such that 
$\stra(w)(a)=1$.
A \emph{memoryless} strategy is independent of the history of the play and
depends only on the current state. 
Memoryless strategies correspond to selectors; we write
$\overline{\xi}$ for the memoryless strategy consisting in playing
forever the selector $\xi$. 
A strategy is \emph{pure memoryless} if it is both pure and
memoryless.
In a turn-based stochastic game, a strategy for player~1 is a function
$\stra_1:S^* \cdot S_1 \to \Distr(S)$, such that for all $w \in S^*$
and for all $s \in S_1$ we have $\supp(\stra_1(w\cdot s)) \subs E(s)$.
Memoryless strategies and pure memoryless strategies are obtained 
as the restriction of strategies as in the case of concurrent game 
graphs.
The family of strategies for player~2 are defined analogously.
We denote by $\Stra_1$ and $\Stra_2$ the sets of all strategies for 
player $1$ and player $2$, respectively.
We denote by $\Stra_i^M$ and $\Stra_i^\PM$ the sets of memoryless strategies 
and pure memoryless strategies for player~$i$, respectively. 

\medskip\noindent{\bf Destinations of moves and selectors.}
For all states $s \in S$ and moves $a_1 \in \mov_1(s)$ and $a_2 \in \mov_2(s)$,
we indicate by $\dest(s,a_1,a_2) = \supp(\trans(s,a_1,a_2))$
the set of possible successors of $s$ when the moves $a_1$ and $a_2$ are 
chosen.
Given a state $s$, and selectors $\xi_1$ and $\xi_2$ for the two players, 
we denote by 
\[
  \dest(s,\xi_1,\xi_2) = \bigcup_{\begin{array}{c}
      \scriptstyle a_1 \in \supp(\xi_1(s)),\\ 
      \scriptstyle a_2 \in \supp(\xi_2(s)) \end{array}} \dest(s,a_1,a_2)
\] 
the set of possible successors of $s$ with respect to the 
selectors $\xi_1$ and $\xi_2$. 

Once a starting state $s$ and strategies $\stra_1$ and $\stra_2$
for the two players are fixed, the game is reduced to an
ordinary stochastic  process.
Hence, the probabilities of events are uniquely defined, where an {\em
event\/} $\cala\subseteq\Paths_s$ is a measurable set of
plays.
For an event $\cala\subseteq\Paths_s$, we denote by
$\Prb_s^{\stra_1,\stra_2}(\cala)$ the probability that a play belongs to
$\cala$ when the game starts from $s$ and the players follows the
strategies $\stra_1$ and~$\stra_2$.
Similarly, for a measurable function $f: \Paths_s \to \reals$, 
we denote by $\E_s^{\stra_1,\stra_2}(f)$ the expected value of $f$ when
the game starts from $s$ and the players follow the strategies
$\stra_1$ and~$\stra_2$.
For $i \geq 0$, we denote by $\randpath_i: \Paths
\to S$ the random  variable denoting the $i$-th state along a
play. 

\medskip\noindent{\bf Valuations.}
A {\em valuation\/} is a mapping $v: S \to [0,1]$ associating a real
number $v(s) \in [0,1]$ with each state $s$. 
Given two valuations $v, w: S \to \reals$, we write $v \leq w$ when
$v(s) \leq w(s)$ for all states $s \in S$. 
For an event $\cala$, we denote by 
$\Prb^{\stra_1,\stra_2}(\cala)$ the valuation $S \to [0,1]$
defined for all states $s \in S$ 
by $\bigl(\Prb^{\stra_1,\stra_2}(\cala)\bigr)(s) =
\Prb_s^{\stra_1,\stra_2}(\cala)$. 
Similarly, for a measurable function $f: \Omega_s \to [0,1]$,
we denote by $\E^{\stra_1,\stra_2}(f)$ the valuation $S \to [0,1]$
defined for all $s \in S$ by $\bigl(\E^{\stra_1,\stra_2}(f)\bigr)(s) =
\E_s^{\stra_1,\stra_2}(f)$.

\medskip\noindent{\bf The $\Pre$ operator.}
Given a valuation $v$, and two selectors 
$\xi_1 \in \Sel_1$ and $\xi_2 \in \Sel_2$, 
we define the valuations 
$\Pre_{\xi_1,\xi_2}(v)$, $\Pre_{1:\xi_1}(v)$, and $\Pre_1(v)$ 
as follows, for all states $s \in S$: 
\begin{multline*} 
  \Pre_{\xi_1,\xi_2}(v)(s) 
  = \sum_{a,b \in \moves} \, \sum_{t \in S} v(t) \cdot \trans(s,a,b)(t) \cdot
    \xi_1(s)(a) \cdot \xi_2(s)(b)
  \\ \quad
  \Pre_{1:\xi_1}(v)(s) 
   =  \inf_{\xi_2 \in \Sel_2} \, \Pre_{\xi_1 ,\xi_2}(v)(s) \hfill
  \\ \quad
  \Pre_1(v)(s) 
   = \sup_{\xi_1 \in \Sel_1} \, \inf_{\xi_2 \in \Sel_2} \, 
  \Pre_{\xi_1 ,\xi_2}(v)(s) \hfill
\end{multline*}
Intuitively, $\Pre_1(v)(s)$ is the greatest expectation of $v$ that
player~1 can guarantee at a successor state of $s$. 
Also note that given a valuation $v$, the computation of $\Pre_1(v)$ 
reduces to the solution of a zero-sum one-shot matrix game, and can be 
solved by linear programming.
Similarly, 
$\Pre_{1:\xi_1}(v)(s)$ is the greatest expectation of $v$ that
player~1 can guarantee at a successor state of $s$ by playing the 
selector $\xi_1$. 
Note that all of these operators on valuations are monotonic: 
for two valuations $v, w$, if $v \leq w$,
then for all selectors $\xi_1 \in \Sel_1$ and $\xi_2 \in \Sel_2$, we have 
$\Pre_{\xi_1,\xi_2}(v) \leq \Pre_{\xi_1,\xi_2}(w)$, 
$\Pre_{1:\xi_1}(v) \leq \Pre_{1:\xi_1}(w)$, 
and $\Pre_1(v) \leq \Pre_1(w)$. 

\medskip\noindent{\bf Reachability and safety objectives.}
Given a set $F \subs S$ of \emph{safe} states, the objective of a safety 
game consists in never leaving $F$.
Therefore, we define the set of winning plays as the set 
$\Safe(F)=\set{\seq{s_0,s_1, s_2,\ldots} \in \Paths \mid s_k \in F 
\mbox{ for all } k \geq 0}$.
Given a subset $T \subs S$ of \emph{target} states, the objective of a
reachability game consists in reaching $T$. 
Correspondingly, the set winning plays is 
$\Reach(T) = \set{\seq{s_0, s_1, s_2,\ldots} \in \Paths \mid s_k
\in T \mbox{ for some }k \ge 0}$ of plays that visit $T$.
For all $F \subs S$ and $T \subs S$, the sets $\Safe(F)$ and $\Reach(T)$ 
is measurable.
An objective in general is a measurable set, and in this paper we 
consider only reachability and safety objectives.
For an objective $\Phi$, the probability of satisfying 
$\Phi$ from a state $s \in S$ under strategies $\stra_1$ and $\stra_2$ 
for players~1 and~2, respectively,
is $\Prb_s^{\stra_1,\stra_2}(\Phi)$.
We define the \emph{value} for player~1 of game with objective $\Phi$ 
from the state $s \in S$ as 
\[
  \va(\Phi)(s) =
  \sup_{\stra_1\in\Stra_1}\inf_{\stra_2\in\Stra_2}
  \Prb_s^{\stra_1,\stra_2}(\Phi); 
\]
i.e., the value is the maximal probability with which player~1 can 
guarantee the satisfaction of $\Phi$ against all player~2 strategies.
Given a player-1 strategy $\stra_1$, we use the notation
\[
\winval{1}^{\stra_1}(\Phi)(s)
= \inf_{\stra_2 \in \Stra_2} \Prb_s^{\stra_1,\stra_2}(\Phi).
\]
A strategy $\stra_1$ for player~1 is {\em optimal\/} for an
objective $\Phi$ if for all 
states $s \in S$, we have 
\[
\winval{1}^{\stra_1}(\Phi)(s)= \va(\Phi)(s). 
\]
For $\vare > 0$, a strategy $\stra_1$ for player~1 is 
{\em $\vare$-optimal\/} if for all states $s \in S$, we have 
\[
\winval{1}^{\stra_1}(\Phi)(s)
  \geq \va(\Phi)(s) - \vare. 
\]
The notion of values and optimal strategies for player~2 are
defined analogously.
Reachability and safety objectives are dual, i.e., we have
$\Reach(T)=\Paths \setminus \Safe(S\setminus T)$.  
The quantitative determinacy result of~\cite{Eve57} ensures that for all
states $s \in S$, we have
\[  
\va(\Safe(F))(s) + \vb(\Reach(S\setminus F))(s)  = 1. 
\]

\section{Markov Decision Processes}\label{sec:mdp}

\noindent
To develop our arguments, we need some facts about one-player versions of
concurrent stochastic games, known as \emph{Markov decision processes}
(MDPs) \cite{Derman,Bertsekas95}. 
For $i \in \set{1,2}$, a \emph{player-$i$ MDP} 
(for short, $i$-MDP) is a concurrent game where, for all
states $s \in S$, we have $|\mov_{3-i}(s)|=1$.
Given a concurrent game $G$, if we fix a memoryless strategy
corresponding to selector $\xi_1$ for player~1, the game
is equivalent to a 2-MDP $G_{\xi_1}$ with the transition function 
\[
\trans_{\xi_1}(s,a_2)(t) = \sum_{a_1 \in \mov_1(s)} 
\trans(s,a_1,a_2)(t) \cdot \xi_1(s)(a_1),
\]
for all $s \in S$ and $a_2 \in \mov_2(s)$. 
Similarly, if we fix selectors $\xi_1$ and $\xi_2$ for both players in a
concurrent game $G$, we obtain a Markov chain, which we denote by
$G_{\xi_1,\xi_2}$. 

\medskip\noindent{\bf End components.}
In an MDP, the sets of states that play an equivalent role to the 
closed recurrent classes of Markov chains \cite[Chapter~4]{Kemeny}  are called
``end components'' \cite{CY95,luca-thesis}. 

\begin{definition}{(End components).}
An {\em end component\/} of an $i$-MDP $G$, for $i\in\set{1,2}$, is a 
subset $C \subs S$ of the states such that there is a selector $\xi$ 
for player~$i$ so that 
$C$ is a closed recurrent class of the Markov chain $G_\xi$. 
\end{definition}

\noindent
It is not difficult to see that an equivalent characterization of an
end component $C$ is the following. 
For each state $s \in C$, there is a subset $M_i(s) \subs \mov_i(s)$
of moves such that:
\begin{enumerate}
\item {\em (closed)} if a move in $M_i(s)$ is chosen 
by player $i$ at state $s$, then all successor states that are obtained 
with nonzero probability lie in $C$; and

\item {\em (recurrent)} the graph $(C,E)$, where $E$
consists of the transitions that occur with nonzero probability when
moves in $M_i(\cdot)$ are chosen by player $i$, is strongly connected.

\end{enumerate}
Given a play $\pat\in\Paths$, 
we denote by $\infi(\pat)$ the set of states that
occurs infinitely often along $\pat$.  
Given a set $\calf \subs 2^S$ of subsets of states, we denote by 
$\infi(\calf)$ the event $\set{\pat \mid \infi(\pat) \in \calf}$.
The following theorem states that in a 2-MDP, for every strategy of
player~2, the set of states that are visited infinitely often is,
with probability~1, an end component. 
Corollary~\ref{coro:prob1} follows easily from Theorem~\ref{theo-ec}.

\begin{theorem}{(\cite{luca-thesis}).} \label{theo-ec}
For a player-1 selector $\xi_1$, 
let $\calc$ be the set of end components of a 2-MDP $G_{\xi_1}$.
For all player-2 strategies $\stra_2$ and all states $s \in S$, we have
$\Prb_s^{\ov{\xi}_1,\stra_2}(\infi(\calc)) = 1$. 
\end{theorem}

\begin{corollary}{}\label{coro:prob1}
For a player-1 selector $\xi_1$, 
let $\calc$ be the set of end components of a 2-MDP $G_{\xi_1}$, and
let $Z= \bigcup_{C \in \calc} C$ be the set of states of all
end components.
For all player-2 strategies $\stra_2$ and all states $s \in S$, we have
$\Prb_s^{\ov{\xi}_1,\stra_2}(\Reach(Z)) = 1$. 
\end{corollary}

\paragraph{MDPs with reachability objectives.}\label{subsec:mdpreach}

Given a 2-MDP with a reachability objective $\Reach(T)$ for player~2,
where $T\subseteq S$, the values can be obtained as the solution of a 
linear program~\cite{FV97} (see Section 2.9 of~\cite{FV97} where linear
program solution is given for MDPs with limit-average objectives and
reachability objective is a special case of limit-average objectives).
The linear program has a variable $x(s)$ for all states $s \in S$, and the
objective function and the constraints are as follows:
\[
\min \ \displaystyle \sum_{s \in S} x(s) \quad \text{subject to } \\
\]
\vspace*{-2ex}
\begin{align*}
 x(s) \geq \displaystyle \sum_{t \in S} x(t) \cdot \trans(s,a_2) (t) 
	& \text{\ \ for all \ } s\in S \text{\ and\ } 
        a_2 \in \mov_2(s) \\
 x(s) = 1 & \text{\ \ for all \ } s \in T \\
 0 \leq x(s) \leq 1 & \text{\ \ for all \ } s \in S
\end{align*}
The correctness of the above linear program to compute the values follows 
from~\cite{FV97} (see section 2.9 of~\cite{FV97}, and also see~\cite{CY95}
for the correctness of the linear program). 

\newcommand{\vars}[1]{\winval{1}^{#1}(\Reach(T))}
\newcommand{\vbrs}[1]{\winval{2}^{#1}(\Safe(S \setminus T))}
\newcommand{\var}{\va(\Reach(T))}

\section{Existence of Memoryless $\vare$-Optimal Strategies for
  Concurrent Reachability Games}

\noindent
In this section we present an elementary and combinatorial proof of the 
existence of memoryless $\vare$-optimal strategies
for concurrent reachability games, for all $\vare > 0$
(optimal strategies need not exist for concurrent
games with reachability objectives~\cite{Eve57}). 

\subsection{From value iteration to selectors} 
\label{sec-valitersel}

\noindent
Consider a reachability game with target $T \subs S$, i.e., objective 
for player~1 is $\Reach(T)$. 
Let $W_2=\set{s \in S \mid \var(s)=0}$ be the set of states from which
player~1 cannot reach the target with positive probability.
From \cite{dAH00}, we know that this set can be computed as $W_2
= \lim_{k \rightarrow \infty} W_2^k$, where $W_2^0 = S \setm T$, and
for all $k \geq 0$,
\[
W_2^{k+1}  = \set{ s \in S \setm T \mid 
    \exists a_2 \in \mov_2(s) \qdot \forall a_1 \in \mov_1(s) \qdot  
    \dest(s,a_1,a_2) \subs W_2^k} \eqpun . 
\]
%
%
The limit is reached in at most $|S|$ iterations. 
Note that player~2 has a strategy that confines the game to $W_2$, 
and that consequently all strategies are optimal for player~1, as they
realize the value~0 of the game in $W_2$. 
Therefore, without loss of generality, in the remainder we assume that
all states in $W_2$ and $T$ are absorbing. 

Our first step towards proving the existence of memoryless $\vare$-optimal
strategies for reachability games consists in considering a
value-iteration scheme for the computation of $\var$. 
Let $[T]: S \to [0,1]$ be the indicator function of $T$, defined
by $[T](s) = 1$ for $s \in T$, and $[T](s) = 0$ for $s \not\in T$. 
Let $u_0=[T]$, and for all $k\ge 0$, let 
\begin{align} \label{eq-valiter} 
  u_{k+1} & = \Pre_1(u_k). 
\end{align}
Note that the classical equation assigns $u_{k+1} = [T] \lor Pre_1(u_k)$,
where $\lor$ is interpreted as the maximum in pointwise fashion.
Since we assume that all states in $T$ are absorbing, the classical
equation reduces to the simpler equation given by (\ref{eq-valiter}).
From the monotonicity of $\Pre_1$ it follows that $u_k \leq u_{k+1}$, that is,
$\Pre_1(u_k) \geq u_k$, for all $k \geq 0$. 
The result of \cite{dAM04} establishes by a combinatorial argument
that $\winval{1}(\Reach(T)) =  \lim_{k \to \infty} u_k$, 
where the limit is interpreted in pointwise fashion. 
For all $k \geq 0$, let the player-1 selector $\zeta_k$ be a 
\emph{value-optimal} selector for $u_k$, that is, a selector such that 
$\Pre_1(u_k) = \Pre_{1:\zeta_k}(u_k)$.
An $\vare$-optimal strategy $\stra_1^k$ for
player~1 can be constructed by applying the sequence $\zeta_k,
\zeta_{k-1}, \ldots, \zeta_1, \zeta_0, \zeta_0, \zeta_0, \ldots$
of selectors,
where the last selector, $\zeta_0$, is repeated forever.  
It is possible to prove by induction on $k$ that 
\[
  \inf_{\stra_2 \in \Stra_2} \Prb^{\stra_1^k,\stra_2} 
  (\exists j \in [0..k] .\, \randpath_j \in T) \geq u_k.
\]
As the strategies $\stra_1^k$, for $k \geq 0$, are not necessarily
memoryless, this proof does not suffice for showing the existence of 
memoryless $\vare$-optimal strategies. 
On the other hand, the following example shows that the memoryless
strategy $\overline{\zeta}_k$ does not necessarily guarantee the value
$u_k$. 

\begin{example}{} \label{examp-nonopt} 
Consider the $1$-MDP shown in Fig~\ref{fig1}.
At all states except $s_3$, the set of available
moves for player~1 is a singleton set. 
At $s_3$, the available moves for player~1 are $a$ and $b$.
The transitions at the various states are shown in
the figure.
The objective of player~1 is to reach the state $s_0$.

We consider the value-iteration procedure
and denote by $u_k$ the valuation after $k$ iterations.
Writing a valuation $u$ as the list of values
$\bigl(u(s_0), u(s_1), \ldots, u(s_4)\bigr)$, we have:
\begin{align*}
u_0 & = (1,0,0,0,0) \\
u_1 = \Pre_1(u_0) & = (1,0,\half,0,0) \\
u_2 = \Pre_1(u_1) & = (1,0,\half,\half,0) \\
u_3 = \Pre_1(u_2) & = (1,0,\half,\half,\half) \\
u_4 = \Pre_1(u_3) = u_3 & = (1,0,\half,\half,\half) 
\end{align*}
The valuation $u_3$ is thus a fixpoint. 

Now consider the selector $\xi_1$ for player~1 that chooses
at state $s_3$ the move $a$ with probability~1.
The selector $\xi_1$ is optimal with respect to the valuation $u_3$.
However, if player~1 follows the memoryless strategy $\ov{\xi}_1$,
then the play visits $s_3$ and $s_4$ alternately and reaches
$s_0$ with probability~0.
Thus, $\xi_1$ is an example of a selector that is value-optimal, but
not optimal. 

On the other hand, consider any selector $\xi'_1$ for
player~1 that chooses move $b$ at state $s_3$ with positive probability.
Under the memoryless strategy $\ov{\xi}'_1$, the set $\set{s_0,s_1}$ of 
states is
reached with probability~1, and $s_0$ is reached with probability
$\half$. 
Such a $\xi'_1$ is thus an example of a selector that is both
value-optimal and optimal. \qed
\begin{figure}[t]
\begin{center}
\begin{picture}(48,28)(15,-5)
\node[Nmarks=n](n0)(40,12){$s_2$}
\node[Nmarks=n](n1)(70,0){$s_0$}
\node[Nmarks=n](n2)(70,24){$s_1$}

\node[Nmarks=n](n3)(20,12){$s_3$}
\node[Nmarks=n](n4)(0,12){$s_4$}

\drawloop[loopangle=0, loopdiam=6](n1){}
\drawloop[loopangle=0, loopdiam=6](n2){}
\drawedge[ELpos=50, ELside=r, ELdist=0.5, curvedepth=0](n0,n1){$1/2$}
\drawedge[ELpos=50, ELside=l, curvedepth=0](n0,n2){$1/2$}
\drawedge[ELpos=50, ELside=l, curvedepth=0](n3,n0){$b$}
\drawedge[ELpos=50, ELside=l, curvedepth=6](n3,n4){$a$}
\drawedge[ELpos=50, ELside=l, curvedepth=6](n4,n3){}
\end{picture}
\end{center}
\caption{An MDP with reachability objective.}\label{fig1}
\end{figure}
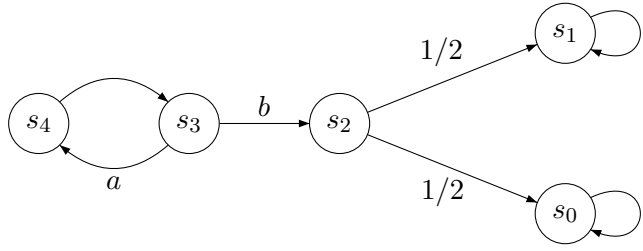
\end{example}

\noindent
In the example, the problem is that the strategy $\overline{\xi}_1$
may cause player~1 to stay forever in $S \setm (T \union W_2)$ with
positive probability. 
We call ``proper'' the strategies of player~1 that guarantee
reaching $T \union W_2$ with probability~1. 

\begin{definition}{(Proper strategies and selectors).}
A player-1 strategy $\stra_1$ is \emph{proper}
if for all player-2 strategies $\stra_2$,
and for all states $s \in S \setminus \wab$, we have 
$\Prb_s^{\stra_1,\stra_2}(\Reach{\wab})=1$.
A player-1 selector $\xi_1$ is \emph{proper} if the memoryless player-1
strategy $\overline{\xi}_1$ is proper. 
\end{definition} 

\noindent
We note that proper strategies are closely related to Condon's
notion of a {\em halting game\/} \cite{Con92}: precisely, a game is
halting iff all player-1 strategies are proper. 
We can check whether a selector for player~1 is proper
by considering only the pure selectors for player~2. 

\begin{lemma}{}
Given a selector $\xi_1$ for player~1, the memoryless player-1 strategy
$\overline{\xi}_1$ is proper iff for every pure selector $\xi_2$ for
player~2, and for all states $s \in S$, we have 
$\Prb_s^{\overline{\xi}_1, \overline{\xi}_2} (\Reach{\wab})=1$.
\end{lemma}
\begin{proof}
We prove the contrapositive. 
Given a player-1 selector $\xi_1$, consider the 2-MDP $G_{\xi_1}$. 
If $\overline{\xi}_1$ is not proper, then by Theorem~\ref{theo-ec},
there must exist an end component $C \subs S \setm \wab$ in $G_{\xi_1}$.
Then, from $C$, player~2 can avoid reaching $T\union W_2$ by 
repeatedly applying 
a pure selector $\xi_2$ that at every state $s \in C$ deterministically 
chooses a move $a_2 \in \mov_2(s)$ such that
$\dest(s,\xi_1,a_2)\subs C$.
The existence of a suitable $\xi_2(s)$ for all states 
$s \in C$ follows from the definition of end component. 
\qed
\end{proof}

\medskip
The following lemma shows that the selector that chooses all 
available moves uniformly at random is proper. 
This fact will be used later to initialize our strategy-improvement
algorithm. 

\begin{lemma}{}\label{lemm:proper2}
Let $\xi_1^\unif$ be the player-1 selector that at all states 
$s \in S \setminus \wab$ chooses all moves in $\mov_1(s)$ uniformly 
at random.
Then $\xi_1^\unif$ is proper.
\end{lemma}
\begin{proof}
Assume towards contradiction that $\xi_1^\unif$ is not proper.
From Theorem~\ref{theo-ec}, in the 2-MDP $G_{\xi_1^\unif}$ there must
be an end component $C \subs S \setm \wab$. 
Then, when player~1 follows the strategy $\overline{\xi}_1^\unif$,
player~2 can confine the game to~$C$.
By the definition of $\xi_1^\unif$, player~2 can ensure that the game does
not leave $C$ regardless of the moves chosen by player~1, and thus, for
{\em all\/} strategies of player~1. 
This contradicts the fact that $W_2$ contains all states from which
player~2 can ensure that $T$ is not reached.
\qed
\end{proof}

\medskip
The following lemma shows that if the player-1 selector $\zeta_k$ computed 
by the value-iteration scheme (\ref{eq-valiter}) is proper, then the 
player-1 strategy $\overline{\zeta}_k$ guarantees the value $u_k$, for 
all $k\ge 0$.

\begin{lemma}{} \label{lem-selector} 
Let $v$ be a valuation such that $\Pre_1(v) \geq v$ and 
$v(s) = 0$ for all states $s \in W_2$.
Let $\xi_1$ be a selector for player~1 such that 
$\Pre_{1:\xi_1}(v) = \Pre_1(v)$.
If $\xi_1$ is proper, then 
for all player-2 strategies $\stra_2$, we have 
$\Prb^{\overline{\xi}_1,\stra_2}(\Reach(T)) \geq v$.
\end{lemma}
\begin{proof}
Consider an arbitrary player-2 strategy $\stra_2$, and for $k \geq 0$, let 
\[
  v_k = \E^{\overline{\xi}_1,\stra_2}\bigl(v(\randpath_k)\bigr)
\]
be the expected value of $v$ after $k$ steps under $\overline{\xi}_1$
and $\stra_2$. 
By induction on $k$, we can prove $v_k \geq v$ for all $k \geq 0$.
In fact, $v_0 = v$, and for $k \geq 0$, we have 
\[
  v_{k+1} \geq \Pre_{1:\xi_1}(v_k) \geq \Pre_{1:\xi_1}(v) =\Pre_1(v)
  \geq v.
\] 
For all $k \geq 0$ and $s \in S$, we can write $v_k$ as
\begin{align*}
  v_k(s) & = \E_s^{\overline{\xi}_1,\stra_2}\bigl(v(\randpath_k) \mid
  \randpath_k \in T \bigr) 
  \cdot \Prb_s^{\overline{\xi}_1,\stra_2}\bigl(\randpath_k \in T\bigr) \\
  & +   
  \E_s^{\overline{\xi}_1,\stra_2}\bigl(v(\randpath_k) \mid
  \randpath_k \in S \setm (T \union W_2) \bigr) 
  \cdot 
 \Prb_s^{\overline{\xi}_1,\stra_2}\bigl(\randpath_k \in S \setm (T
  \union W_2) \bigr) 
 \\
  & + \E_s^{\overline{\xi}_1,\stra_2}\bigl(v(\randpath_k) \mid
  \randpath_k \in W_2 \bigr) 
  \cdot \Prb_s^{\overline{\xi}_1,\stra_2}\bigl(\randpath_k \in W_2\bigr).
\end{align*}
Since $v(s) \leq 1$ when $s \in T$, the first term on the right-hand
side is at most 
$\Prb_s^{\overline{\xi}_1,\stra_2}\bigl(\randpath_k \in T\bigr)$.
For the second term, we have 
$\lim_{k \to \infty} \Prb^{\overline{\xi}_1,\stra_2}\bigl(\randpath_k
\in S \setm (T \union W_2) \bigr) = 0$ by hypothesis, because
$\Prb^{\overline{\xi}_1,\stra_2}(\Reach(T \union W_2)) = 1$ and 
every state $s \in (T \union W_2)$ is absorbing.
Finally, the third term on the right hand side is~0, as $v(s) = 0$ for all 
states $s \in W_2$. 
Hence, taking the limit with $k \to \infty$, we obtain 
\begin{align*}
  \Prb^{\overline{\xi}_1,\stra_2}\bigl(\Reach(T)\bigr)
  & =
  \lim_{k \to \infty} 
  \Prb^{\overline{\xi}_1,\stra_2}\bigl(\randpath_k \in T\bigr) 
  \geq 
  \lim_{k \to \infty} v_k 
  \geq v,
\end{align*}
where the last inequality follows from $v_k \geq v$ for all $k \geq 0$. 
Note that $v_k= \Prb^{\overline{\xi}_1,\stra_2}\bigl(\randpath_k \in T\bigr)$,
and since $T$ is absorbing it follows that $v_k$ is non-deccreasing (monotonic)
and is bounded by~1 (since it is a probability measure).
Hence the limit of $v_k$ is defined.  
The desired result follows.
\qed
\end{proof}

\subsection{From value iteration to optimal selectors} 
\label{sec-optisel} 

\noindent
In this section we show how to obtain memoryless $\vare$-optimal strategies 
from the value-iteration scheme, for $\vare>0$.
In the following section the existence such strategies would be established
using a strategy-iteration scheme. 
The strategy-iteration scheme has been used previously to establish existence
of memoryless $\vare$-optimal strategies, for $\vare>0$ 
(for example see~\cite{EY06} and also results of Condon~\cite{Con92} for turn-based games).
However our proof which constructs the memoryless strategies based on 
value-iteration scheme is new.
Considering again the value-iteration scheme (\ref{eq-valiter}), 
since $\var = \lim_{k \to \infty} u_k$, for every $\vare>0$ there 
is a $k$ such that $u_k(s) \geq u_{k-1}(s) \geq \var(s) -\vare$ at all 
states $s \in S$.  
Lemma~\ref{lem-selector} indicates that, in order to construct a
memoryless $\vare$-optimal strategy, we need to construct from $u_{k-1}$ a
player-1 selector $\xi_1$ such that: 
\begin{enumerate}

\item $\xi_1$ is value-optimal for $u_{k-1}$, that is,
  $\Pre_{1:\xi_1}(u_{k-1}) = \Pre_1(u_{k-1})=u_k$; and
\item $\xi_1$ is proper. 

\end{enumerate}
To ensure the construction of a value-optimal, proper selector, we
need some definitions. 
For $r>0$, the \emph{value class} 
\[
  U^k_r =\set{s \in S \mid u_{k}(s) =r }
\]
consists of the states with value $r$ under the valuation $u_{k}$. 
Similarly we define $U^k_{\bowtie r} =\set{s \in S \mid u_{k}(s) \bowtie r}$, 
for $\bowtie \, \in \set{<,\leq,\geq,>}$. 
For a state $s \in S$, let 
$\ell_k(s) = \min\set{j \leq k \mid u_j(s) = u_k(s)}$ be the 
\emph{entry time} of $s$ in $U^k_{u_k(s)}$, 
that is, the least iteration $j$ in which the state $s$ has the same
value as in iteration $k$. 
For $k \geq 0$, we define the player-1 selector $\eta_k$ as follows:
if $\ell_k(s) > 0$, then 
\[
  \eta_k(s) = \eta_{\ell_k(s)}(s) = \arg \max_{\xi_1 \in \Sel_1} 
  \inf_{\xi_2 \in \Sel_2} \Pre_{\xi_1 ,\xi_2}(u_{\ell_k(s)-1});
\]
otherwise, if $\ell_k(s) = 0$, then $\eta_k(s) = \eta_{\ell_k(s)}(s) =
  \xi_1^\unif(s)$ 
(this definition is arbitrary, and it does not affect
the remainder of the proof). 
In words, the selector $\eta_k(s)$ is an optimal selector for $s$ at
the iteration $\ell_k(s)$.
It follows easily that $u_k=\Pre_{1:\eta_k}(u_{k-1}) $, that is, 
$\eta_k$ is also value-optimal for $u_{k-1}$, satisfying
the first of the above conditions. 

To conclude the construction, we need to prove that for $k$
sufficiently large (namely, for $k$ such that $u_k(s) > 0$ at all 
states $s\in S \setm \wab$), the selector $\eta_k$ is proper. 
To this end we use Theorem~\ref{theo-ec}, and show that for
sufficiently large $k$ no end component of $G_{\eta_k}$ is entirely 
contained in $S \setm \wab$.\footnote{%
  In fact, the result holds for all $k$,
  even though our proof, for the sake of a simpler argument, does not
  show it.}
To reason about the end components of $G_{\eta_k}$, 
for a state $s \in S$ and a player-2 move $a_2 \in \mov_2(s)$, 
we write
\[
\dest_k(s,a_2) =\bigcup_{a_1 \in \supp(\eta_k(s))} \dest(s,a_1,a_2)
\]
for the set of possible successors of state $s$ when player~1 follows the 
strategy $\ov{\eta}_k$, and player~2 chooses the move $a_2$. 

\begin{lemma}{}\label{lem-1}
Let $0 < r \leq 1$ and $k \geq 0$,
and consider a state $s \in S \setminus (T \union W_2)$
such that $s \in U^k_r$.
For all moves $a_2 \in \mov_2(s)$, we have: 
\begin{enumerate} 
\item either $\dest_k(s,a_2) \inters U^k_{> r} \neq \emptyset$,
\item or $\dest_k(s,a_2) \subs U^k_r$, and there is a state
  $t \in \dest_k(s,a_2)$ with $\ell_k(t) < \ell_k(s)$. 
\end{enumerate}
\end{lemma}
\begin{proof}
For convenience, let $m = \ell_k(s)$, and 
consider any move $a_2 \in \mov_2(s)$. 
\begin{itemize} 
\item Consider first the case that $\dest_k(s,a_2) \not\subs U^k_r$.
Then, it cannot be that $\dest_k(s,a_2) \subs U^k_{\leq r}$;
otherwise, for all states $t \in \dest_k(s,a_2)$, we would have 
$u_k(t) \leq r$, and there would be at least one state 
$t \in \dest_k(s,a_2)$
such that $u_k(t) < r$, contradicting $u_k(s) = r$ and
$\Pre_{1:\eta_k}(u_{k-1}) = u_k$. 
So, it must be that $\dest_k(s,a_2) \inters U^k_{>r} \neq \emptyset$. 

\item Consider now the case that $\dest_k(s,a_2) \subs U^k_r$.
Since $u_m \leq u_k$, due to the monotonicity of the
$\Pre_1$ operator and (\ref{eq-valiter}), we have that 
$u_{m-1}(t) \leq r$ for all states $t \in \dest_k(s,a_2)$.
From $r = u_k(s) = u_m(s) = \Pre_{1:\eta_k}(u_{m-1})$, 
it follows that $u_{m-1}(t) = r$ for all states $t \in \dest_k(s,a_2)$,
implying that $\ell_k(t) < m$ for all states $t \in \dest_k(s,a_2)$. 
\qed
\end{itemize}
\end{proof}

\medskip
The above lemma states that under $\eta_k$, from
each state $i \in U^k_r$ with $r > 0$ we are guaranteed a probability
bounded away from~0 of either moving to a higher-value class
$U^k_{>r}$, or of moving to states within the value class that have a
strictly lower entry time. 
Note that the states in the target set $T$ are all in $U^0_1$: they 
have entry-time~0 in the value class for value~1. 
This implies that every state in $S \setm W_2$ has a probability
bounded above zero of reaching $T$ in at most $n = |S|$ steps, so that
the probability of staying forever in $S \setm (T \union W_2)$ is~0.
To prove this fact formally, we analyze the end components of
$G_{\eta_k}$ in light of Lemma~\ref{lem-1}. 

\begin{lemma}{} \label{lem-2}
For all $k \geq 0$, if for all states 
$s \in S \setm W_2$ we have $u_{k-1}(s)>0$,  
then for all player-2 strategies $\stra_2$, we have 
$\Prb^{\overline{\eta}_k,\stra_2}\bigl(\Reach (T \union W_2)) = 1$. 
\end{lemma}
\begin{proof}
Since every state $s \in (T \union W_2)$ is absorbing,
to prove this result, in view of Corollary~\ref{coro:prob1}, it suffices to
show that no end component of $G_{\eta_k}$ is entirely contained
in $S \setm (T \union W_2)$. 
Towards the contradiction, assume there is such an end component 
$C \subs S \setm (T \union W_2)$. 
Then, we have 
$C \subs U^k_{[r_1,r_2]}$ with $C \inters U_{r_2} \neq \emptyset$, 
for some $0 < r_1 \leq r_2 \leq 1$,
where $U^k_{[r_1,r_2]} = U^k_{\geq r_1} \inters U^k_{\leq r_2}$ is the
union of the value classes for all values in the interval $[r_1,r_2]$. 
Consider a state $s \in U^k_{r_2}$ with minimal $\ell_k$, that is, such
that $\ell_k(s) \leq \ell_k(t)$ for all other states $t \in U^k_{r_2}$.
From Lemma~\ref{lem-1}, it follows that for every move $a_2 \in \mov_2(s)$,
there is a state $t \in \dest_k(s,a_2)$ such that 
(i)~either $t \in U^k_{r_2}$ and $\ell_k(t) < \ell_k(s)$, 
(ii)~or $t \in U^k_{>r_2}$. 
In both cases, we obtain a contradiction. 
\qed
\end{proof}

\medskip
The above lemma shows that $\eta_k$ satisfies both requirements
for optimal selectors spelt out at the beginning of
Section~\ref{sec-optisel}.
Hence, $\eta_k$ guarantees the value $u_k$.
This proves the existence of memoryless $\vare$-optimal strategies for
concurrent reachability games. 

\begin{theorem}{(Memoryless $\vare$-optimal strategies).}
For every $\vare>0$, memoryless $\vare$-optimal strategies exist for all
concurrent games with reachability objectives.
\end{theorem}
\begin{proof} 
Consider a concurrent reachability game with target $T \subs S$. 
Since $\lim_{k \to \infty} u_k = \var$, 
for every $\vare > 0$ we can find $k \in \nats$ such that the 
following two assertions hold:
\begin{gather*} 
  \max_{s \in S} \bigl( \var(s) - u_{k-1}(s) \bigr) < \vare \\
  \min_{s \in S\setminus W_2} u_{k-1}(s) > 0
\end{gather*}
By construction, $\Pre_{1:\eta_k}(u_{k-1}) = \Pre_1(u_{k-1})=u_k$.
Hence, from Lemma~\ref{lem-selector} and Lemma~\ref{lem-2}, 
for all player-2 strategies $\stra_2$, we have 
$\Prb^{\overline{\eta}_k,\stra_2}(\Reach(T)) \geq u_{k-1}$, 
leading to the result. 
\qed
\end{proof}

\section{Strategy Improvement Algorithm for Concurrent Reachability Games}\label{sec-stra-improve-reach}

\noindent
In the previous section, we provided a proof of the existence of
memoryless $\vare$-optimal strategies for all $\vare > 0$, on the
basis of a value-iteration scheme. 
In this section we present a strategy-improvement 
algorithm for concurrent games with reachability objectives. 
The algorithm will produce a sequence of selectors 
$\gamma_0, \gamma_1, \gamma_2, \ldots$ for player 1, such that: 
\begin{enumerate}
  \item \label{l-improve-1} for all $i \geq 0$, we have 
  $\vars{\overline{\gamma}_i} \leq \vars{\overline{\gamma}_{i+1}}$;

  \item \label{l-improve-3} 
  if there is $i \geq 0$ such that $\gamma_i = \gamma_{i+1}$, 
  then $\vars{\overline{\gamma}_i} = \var$; and 

  \item \label{l-improve-2} 
  $\lim_{i \rightarrow \infty} \vars{\overline{\gamma}_i} = \var$. 
\end{enumerate}
Condition~\ref{l-improve-1} guarantees that the algorithm computes a
sequence of monotonically improving selectors. 
Condition~\ref{l-improve-3} guarantees that if a selector cannot be
improved, then it is optimal. 
Condition~\ref{l-improve-2} guarantees that the value guaranteed by
the selectors converges to the value of the game, or equivalently,
that for all $\vare > 0$, there is a number $i$ of iterations
such that the memoryless player-1 strategy $\ov{\gamma}_i$ is 
$\vare$-optimal. 
Note that for concurrent reachability games, there may be no $i \geq
0$ such that $\gamma_i = \gamma_{i+1}$, that is, the algorithm may fail to
generate an optimal selector. 
This is because there are concurrent reachability games that do not
admit optimal strategies, but only $\vare$-optimal strategies for all
$\vare > 0$ \cite{Eve57,crg-tcs07}.
For {\em turn-based\/} reachability games, our algorithm terminates with 
an optimal selector and we will present bounds for termination.

We note that the value-iteration scheme of the previous section does
not directly yield a strategy-improvement algorithm. 
In fact, the sequence of player-1 selectors $\eta_0, \eta_1, \eta_2, \ldots$
computed in Section~\ref{sec-valitersel} may violate
Condition~\ref{l-improve-3}: it is possible that for some $i \geq 0$
we have $\eta_i = \eta_{i+1}$, but $\eta_{i} \neq \eta_{j}$ for some
$j > i$. 
This is because the scheme of Section~\ref{sec-valitersel} is
fundamentally a value-iteration scheme, even though a selector is
extracted from each valuation. 
The scheme guarantees that the valuations $u_0, u_1, u_2, \ldots$
defined as in (\ref{eq-valiter}) converge, but it does not guarantee
that the selectors $\eta_0, \eta_1, \eta_2, \ldots$ improve at each
iteration. 

The strategy-improvement algorithm presented here shares an important
connection with the proof of the existence of memoryless $\vare$-optimal
strategies presented in the previous section. 
Here, also, the key is to ensure that all generated selectors are proper.
Again, this is ensured by modifying the selectors, at each
iteration, only where they can be improved.

\subsection{The strategy-improvement algorithm}\label{subsec:reach-stra}

\smallskip\noindent{\bf Ordering of strategies.}
We let $W_2$ be as in Section~\ref{sec-valitersel}, and again we
assume without loss of generality that all states in $W_2 \union T$
are absorbing. 
We define a preorder $\prec$ on the strategies for player 1 as follows:
given two player 1 strategies $\stra_1$ and $\stra_1'$, let
$\stra_1 \prec \stra_1'$ if the following two conditions hold:
(i)~$\vars{\stra_1} \leq \vars{\stra_1'}$; and 
(ii)~$\vars{\stra_1}(s) < \vars{\stra_1'}(s)$ for some state $s\in S$.
Furthermore, we write 
$\stra_1 \preceq \stra_1'$ if either $\stra_1 \prec \stra_1'$ or 
$\stra_1 = \stra_1'$.

\medskip\noindent{\bf Informal description of Algorithm~\ref{algorithm:strategy-improve}.}
We now present the strategy-improvement algorithm 
(Algorithm~\ref{algorithm:strategy-improve}) for computing the values for 
all states in $S \setminus \wab$.
The algorithm iteratively improves player-1 strategies according to the 
preorder $\prec$. 
The algorithm starts with the random selector 
$\gamma_0=\overline{\xi}_1^\unif$.
At iteration $i+1$, the algorithm considers the memoryless player-1 strategy
$\overline{\gamma}_i$ and computes the value $\vars{\overline{\gamma}_i}$.
Observe that since $\overline{\gamma}_i$ is a memoryless strategy, the
computation of $\vars{\overline{\gamma}_i}$ involves solving the 2-MDP 
$G_{{\gamma}_i}$.
The valuation $\vars{\overline{\gamma}_i}$ is named $v_i$.
For all states $s$ such that $\Pre_1(v_i)(s) > v_i(s)$, 
the memoryless strategy at $s$ is modified to a selector 
that is value-optimal for $v_i$. 
The algorithm then proceeds to the next iteration.
If $\Pre_1(v_i) = v_i$, the algorithm stops and returns the optimal 
memoryless strategy $\overline{\gamma}_i$ for player~1.
Unlike strategy-improvement algorithms for turn-based games (see
\cite{Con93} for a survey),
Algorithm~\ref{algorithm:strategy-improve} is not guaranteed to
terminate, because the value of a reachability game may not be rational.


\begin{algorithm*}[t]
\caption{Reachability Strategy-Improvement Algorithm}
\label{algorithm:strategy-improve}
{
\begin{tabbing}
aaa \= aaa \= aaa \= aaa \= aaa \= aaa \= aaa \= aaa \kill
\\
\> {\bf Input:} a concurrent game structure $G$ with target set $T$. \\
\>   {\bf Output:} a strategy $\overline{\gamma}$ for player~1. \\ \\

\> 0. Compute $W_2=\set{s \in S \mid \var(s)=0}$. \\
\> 1. Let $\gamma_0=\xi_1^\unif$ and $i=0$. \\
\> 2. Compute $v_0 = \vars{\overline{\gamma}_0}$. \\

\> 3. {\bf do \{ } \\ 
\>\> 3.1. Let $I= \set{s \in S \setminus \wab \mid \Pre_1(v_i)(s) > v_i(s)}$. \\
\>\> 3.2. Let $\xi_1$ be a player-1 
        selector such that for all states $s \in I$, \\
\>\>\>	we have $\Pre_{1:\xi_1}(v_i)(s) =\Pre_1(v_i)(s) > v_i(s)$.\\
\>\> 3.3. The player-1 selector $\gamma_{i+1}$ is defined
          as follows: for each state $s\in S$, let\\
\>\>\>\> $	\displaystyle 
	\gamma_{i+1}(s)=
	\begin{cases}
	\gamma_i(s) & \text{\ if \ }s\not\in I;\\
	\xi_1(s) & \text{\ if\ }s\in I.
	\end{cases}$
	\\ \\ 
\>\> 3.4. Compute $v_{i+1} =\vars{\overline{\gamma}_{i+1}}$. \\
\>\> 3.5. Let $i=i+1$. \\
\> {\bf \} until } $I=\emptyset$. \\
\> 4. {\bf return} $\overline{\gamma}_{i}$.  
\end{tabbing}
}
\end{algorithm*}

\subsection{Convergence}

\begin{lemma}{}\label{lemm:proper3}
Let $\gamma_i$ and $\gamma_{i+1}$ be the player-1 
selectors obtained at iterations 
$i$ and $i+1$ of Algorithm~\ref{algorithm:strategy-improve}.
If $\gamma_i$ is proper, then $\gamma_{i+1}$ is also proper.
\end{lemma}
\begin{proof}
Assume towards a contradiction that $\gamma_i$ is proper and 
$\gamma_{i+1}$ is not.
Let $\xi_2$ be a pure selector for player~2 
to witness that $\gamma_{i+1}$ is not proper. 
Then there exist a subset $C \subseteq S \setminus \wab$ 
such that $C$ is a closed 
recurrent set of states in the Markov chain $G_{\gamma_{i+1},\xi_2}$.
Let $I$ be the nonempty set of states where the selector is modified to 
obtain $\gamma_{i+1}$ from $\gamma_i$; at all other states
$\gamma_i$ and $\gamma_{i+1}$ agree.

Since $\gamma_i$ and $\gamma_{i+1}$ agree at all states other than the
states in $I$, and $\gamma_i$ is a proper strategy, it follows that 
$C \inters I \neq \emptyset$.
Let 
$U_r^i =\set{s \in S \setminus \wab \mid \vars{\overline{\gamma}_i}(s)=v_i(s)=r}$
be the value class with value $r$ at iteration $i$.
For a state $s \in U_r^i$ the following assertion holds: 
if $\dest(s,\gamma_i,\xi_2) \subsetneq U_r^i$, then 
$\dest(s,\gamma_i,\xi_2) \inters U_{>r}^i \neq \emptyset$.
Let $z =\max\set{r \mid U_r^i \inters C \neq \emptyset}$, that is,
$U_z^i$ is the greatest value class at iteration $i$ with a nonempty 
intersection with the closed recurrent set $C$.
It easily follows that $0 < z < 1$. 
Consider any state $s \in I$, and let $s \in U_q^i$. 
Since $\Pre_1(v_i)(s) > v_i(s)$, it follows that 
$\dest(s,\gamma_{i+1},\xi_2) \inters U_{>q}^i \neq \emptyset$.
Hence we must have $z > q$, and therefore 
$I \inters C \inters U_z^i =\emptyset$.
Thus, for all states 
$s \in U_z^i \inters C$, we have $\gamma_i(s)=\gamma_{i+1}(s)$.
Recall that $z$ is the greatest value class at iteration $i$ with 
a nonempty 
intersection with $C$; hence $U_{>z}^i \inters C=\emptyset$.
Thus for all states $s \in C \inters U_z^i$, we have 
$\dest(s,\gamma_{i+1},\xi_2) \subseteq U_z^i \inters C$.
It follows that $C \subseteq U_z^i$.
However, this gives us three statements that together form a
contradiction: $C \inters I \neq \emptyset$ (or else $\gamma_i$ would not
have been proper), $I \inters C \inters U_z^i =\emptyset$, and $C
\subseteq U_z^i$.
\qed
\end{proof}

\begin{lemma}{}\label{lemm:proper4}
For all $i\geq 0$, the player-1 
selector $\gamma_i$ obtained at iteration $i$ of 
Algorithm~\ref{algorithm:strategy-improve} is proper. 
\end{lemma}
\begin{proof}
By Lemma~\ref{lemm:proper2} we have that $\gamma_0$ is proper.
The result then follows from Lemma~\ref{lemm:proper3} and induction.
\qed
\end{proof}

\begin{lemma}{}\label{lemm:stra-improve}
Let $\gamma_i$ and $\gamma_{i+1}$ be the player-1 selectors obtained at 
iterations $i$ and $i+1$ of Algorithm~\ref{algorithm:strategy-improve}.
Let $I=\set{s \in S \mid \Pre_1(v_i)(s) > v_i(s)}$.
Let $v_i=\vars{\overline{\gamma}_i}$ and 
$v_{i+1}=\vars{\overline{\gamma}_{i+1}}$.
Then $v_{i+1}(s)  \geq  \Pre_1(v_i)(s)$ for all states $s\in S$;
and therefore
$v_{i+1}(s) \geq v_i(s)$ for all states $s\in S$, 
and $v_{i+1}(s) > v_i(s)$ for all states $s\in I$.
\end{lemma}
\begin{proof}
Consider the valuations $v_i$ and $v_{i+1}$ obtained at iterations $i$ and 
$i+1$, respectively, and let $w_i$ be the valuation defined by 
$w_i(s) = 1 - v_i(s)$ for all states $s \in S$. 
Since $\gamma_{i+1}$ is proper (by Lemma~\ref{lemm:proper4}),
it follows that the counter-optimal strategy for player~2 to minimize
$v_{i+1}$ is obtained by maximizing the probability to reach $W_2$.
In fact, there are no end components in $S \setminus (W_2 \union T)$
in the 2-MDP $G_{\gamma_{i+1}}$. 
Let 
\[
  \wh{w}_{i}(s) =
  \begin{cases} 
    w_i(s) & \text{\ if\ }s \in S \setminus I; \\
    1-\Pre_1(v_i)(s) < w_i(s) &\text{\ if\ }s \in I. 
  \end{cases}
\]
In other words, $\wh{w}_{i} = 1 - \Pre_1(v_i)$, and we also have 
$\wh{w}_{i} \leq w_i$.
We now show that $\wh{w}_{i}$ is a feasible solution to the 
linear program for MDPs with the objective $\Reach(W_2)$, as
described in Section~\ref{sec:mdp}.
Since $v_i =\vars{\overline{\gamma}_i}$, it follows that
for all states $s\in S$ and all moves $a_2 \in \mov_2(s)$, we have
\[
  w_i(s) \geq \sum_{t \in S} w_i(t) \cdot \trans_{\gamma_i}(s,a_2).  
\]
For all states $s \in S \setminus I$, 
we have $\gamma_i(s)=\gamma_{i+1}(s)$ and
$\wh{w}_{i}(s) = w_i(s)$,
and since $\wh{w}_{i} \leq w_i$, it follows that for all 
states $s \in S \setminus I$ and all moves $a_2 \in \mov_2(s)$, we have
\[
  \wh{w}_{i}(s) \geq \sum_{t \in S} \wh{w}_{i}(t) \cdot \trans_{\gamma_{i+1}}(s,a_2)
 \qquad \text{ (for $s \in (S\setminus I)$)}.
\]

Since for $s \in I$ the selector $\gamma_{i+1}(s)$ is obtained as an
optimal selector for $\Pre_1(v_i)(s)$, it follows that 
for all states $s\in I$ and all moves $a_2 \in \mov_2(s)$, we have
\[
\Pre_{\gamma_{i+1},a_2}(v_i)(s) \geq \Pre_1(v_i)(s);
\]
in other words, $1-\Pre_1(v_i)(s) \geq 1-\Pre_{\gamma_{i+1},a_2}(v_i)(s)$.
Hence for all states $s\in I$ and all moves $a_2 \in \mov_2(s)$, we have
\[
  \wh{w}_{i}(s) \geq \sum_{t \in S} w_{i}(t) \cdot \trans_{\gamma_{i+1}}(s,a_2). 
\]
Since $\wh{w}_{i} \leq w_i$,  
for all states $s\in I$ and all moves $a_2 \in \mov_2(s)$, we have
\[
  \wh{w}_{i}(s) \geq \sum_{t \in S} \wh{w}_{i}(t) \cdot \trans_{\gamma_{i+1}}(s,a_2) 
\qquad \text{( for  $s\in I$)}.
\] 
Hence it follows that $\wh{w}_{i}$ is a feasible solution to the 
linear program for MDPs with reachability objectives.
Since the reachability valuation 
for player~2 for $\Reach(W_2)$ is the least
solution (observe that the objective function of the linear program is a 
minimizing function), it follows that 
$v_{i+1} \geq 1-\wh{w}_{i} = \Pre_1(v_i)$.
Thus we obtain
$v_{i+1}(s) \geq v_i(s)$ for all states $s \in S$, and 
$v_{i+1}(s) > v_i(s)$ for all states $s \in I$.
\qed
\end{proof}

\begin{theorem}{(Strategy improvement).}\label{thrm:stra-improve}
The following two assertions hold 
about Algorithm~\ref{algorithm:strategy-improve}:
\begin{enumerate}
\item 
For all $i \geq 0$, we have $\overline{\gamma}_i \preceq
\overline{\gamma}_{i+1}$; moreover, 
if $\overline{\gamma}_i = \overline{\gamma}_{i+1}$, then 
$\overline{\gamma}_i$ is an optimal strategy.

\item $\lim_{i \to \infty} v_i =
\lim_{i \to \infty} \vars{\overline{\gamma}_i} = \var$.

\end{enumerate}
\end{theorem}
\begin{proof}
We prove the two parts as follows.
\begin{enumerate}
\item 
The assertion that $\overline{\gamma}_i \preceq \overline{\gamma}_{i+1}$ 
follows from Lemma~\ref{lemm:stra-improve}.
If $\overline{\gamma}_i =\overline{\gamma}_{i+1}$, then 
$\Pre_1(v_i)=v_i$.
Let $v= \var$, and since $v$ is the least solution to satisfy 
$\Pre_1(x)=x$ (i.e., the least fixpoint)~\cite{dAM04}, it follows
that $v_i \geq v$.
From Lemma~\ref{lemm:proper4} it follows that 
$\overline{\gamma}_i$ is proper.
Since $\overline{\gamma}_i$ is proper by Lemma~\ref{lem-selector}, we have 
$\vars{\overline{\gamma}_i} \geq v_i \geq v$.
It follows that $\overline{\gamma}_i$ is optimal for player~1. 

\item  Let $v_0= [T]$ and $u_0=[T]$.
We have $u_0 \leq v_0$.
For all $k\geq 0$, by Lemma~\ref{lemm:stra-improve}, we have 
$v_{k+1} \geq [T] \lor \Pre_1(v_k)$.
For all $k \geq 0$, let $u_{k+1}=[T] \lor \Pre_1(u_k)$.
By induction we conclude that for all $k \geq 0$, we have $u_k \leq v_k$.
Moreover, $v_k \leq \var$, 
that is, for all $k\geq 0$, we have 
$$u_k \leq v_k \leq \var.$$
Since $\lim_{k \to \infty} u_k = \var$, it follows that
\[
\lim_{k \to \infty} \vars{\overline{\gamma}_k}  = 
\lim_{k \to \infty} v_k  = \var.
\]
%
\end{enumerate}
The theorem follows.
\qed
\end{proof}

\subsection{Termination for turn-based stochastic games}
If the input game structure to Algorithm~\ref{algorithm:strategy-improve}
is a turn-based stochastic game structure, then if we start with a 
proper selector $\gamma_0$ that is pure, then for all $i \geq 0$ we can 
choose the selector $\gamma_i$ such that $\gamma_i$ is both proper and pure: 
the above claim follows since given a valuation $v$, if a state $s$ is a 
player~1 state, then there is an action $a$ at $s$ (or choice of an 
edge at $s$) that achieves $\Pre_1(v)(s)$ at $s$.
Since the number of pure selectors is bounded, if we start with a 
pure, proper selector then termination is ensured.
Hence we present a procedure to compute a pure, proper selector,
and then present termination bounds (i.e., bounds on $i$ such 
that $u_{i+1}=u_i$).
The construction of a pure, proper selector is based on the notion 
of \emph{attractors} defined below.

\medskip\noindent{\em Attractor strategy.}
Let $A_0=W_2 \cup T$,  and for $i\geq 0$ we have
\[
A_{i+1}= A_i \cup \set{s \in S_1 \cup S_R \mid E(s) \cap A_i \neq \emptyset}
\cup \set{s \in S_2 \mid E(s) \subs A_i}.
\]
Since for all $s \in S \setminus W_2$ we have $\va(\Reach(T))>0$,
it follows that from all states in $S\setminus W_2$ player~1
can ensure that $T$ is reached with positive probability.
It follows that for some $i \geq 0$ we have $A_i=S$.
The pure \emph{attractor} selector $\xi^*$ is as follows:
for a state $s \in (A_{i+1}\setminus A_i) \cap S_1$ 
we have $\xi^*(s)(t)=1$, where $t \in A_i$ (such a $t$ 
exists by construction).
The pure memoryless strategy $\ov{\xi^*}$ ensures that for 
all $i\geq 0$, from $A_{i+1}$ the game reaches $A_i$ 
with positive probability.
Hence there is no end-component $C$ contained in 
$S\setminus (W_2 \cup T)$ in the MDP $G_{\ov{\xi^*}}$.
It follows that $\xi^*$ is a pure selector that is proper,
and the selector $\xi^*$ can be computed in $O(|E|)$ time.
We now present the termination bounds.

\medskip\noindent{\em Termination bounds.}
We present termination bounds for binary turn-based 
stochastic games.
A turn-based stochastic game is binary if for all $s\in S_R$
we have $|E(s)|\leq 2$, and for all $s \in S_R$ if $|E(s)|=2$,
then for all $t \in E(s)$ we have $\trans(s)(t)=\frac{1}{2}$,
i.e., for all probabilistic states there are at most two
successors and the transition function $\trans$ is uniform.

\begin{lemma}{}\label{lemm:MC-bound}
Let $G$ be a binary Markov chain with $|S|$ states with a reachability
objective $\Reach(T)$.
Then for all $s \in S$ we have 
$\va(\Reach(T))=\frac{p}{q}$, with $p,q \in \nats$ and $p,q \leq 4^{|S|-1}$. 
\end{lemma}
\begin{proof}
The results follow as a special case of Lemma~2 of~\cite{Con93}.
Lemma~2 of~\cite{Con93} holds for halting turn-based stochastic games,
and since Markov chains reaches the set of closed connected recurrent
states with probability~1 from all states the result follows.
\qed
\end{proof}

\begin{lemma}{}\label{lemm:TB-bound}
Let $G$ be a binary turn-based stochastic game with a reachability
objective $\Reach(T)$.
Then for all $s \in S$ we have 
$\va(\Reach(T))=\frac{p}{q}$, with $p,q \in \nats$ and $p,q \leq 4^{|S_R|-1}$. 
\end{lemma}
\begin{proof}
Since pure memoryless optimal strategies exist for both players 
(existence of pure memoryless optimal strategies for both players in
turn-based stochastic reachability games follows from~\cite{Con92}),
we fix pure memoryless optimal strategies $\stra_1$ and 
$\stra_2$ for both players.
The Markov chain $G_{\stra_1,\stra_2}$ can be then reduced to an equivalent 
Markov chains with $|S_R|$ states (since we fix deterministic successors
for states in $S_1 \cup S_2$, they can be collapsed to their successors).
The result then follows from Lemma~\ref{lemm:MC-bound}.
\qed
\end{proof}

From Lemma~\ref{lemm:TB-bound} it follows that at iteration~$i$ of the
reachability strategy improvement algorithm either 
the sum of the values either increases by $\frac{1}{4^{|S_R|-1}}$ or else
there is a valuation $u_i$ such that $u_{i+1}=u_i$.
Since the sum of values of all states can be at most $|S|$, it follows
that algorithm terminates in at most $|S| \cdot 4^{|S_R|-1}$ iterations.
Moreover, since the number of pure memoryless strategies is at most
$\prod_{s \in S_1} |E(s)|$, the algorithm terminates in at most
$\prod_{s \in S_1} |E(s)|$ iterations.
It follows from the results of~\cite{ZP95} that a turn-based stochastic
game structure $G$ can be reduced to a equivalent binary turn-based
stochastic game structure $G'$ such that the set of player~1 and player~2
states in $G$ and $G'$ are the same and the number of probabilistic
states in $G'$ is $O(|\trans|)$, where $|\trans|$ is the size of the
transition function in $G$.
Thus we obtain the following result.

\begin{theorem}{}
Let $G$ be a turn-based stochastic game with a reachability objective 
$\Reach(T)$, then the reachability strategy improvement algorithm
computes the values in time 
\[
O\big(\min\set{\prod_{s \in S_1} |E(s)|, 2^{O(|\trans|)} } \cdot \mathit{poly}(|G|\big);
\]
where $\mathit{poly}$ is polynomial function.
\end{theorem} 

The results of~\cite{GH07} presented an algorithm for turn-based 
stochastic games that works in time $O(|S_R| ! \cdot \mathit{poly}(|G|))$.
The algorithm of~\cite{GH07} works only for turn-based stochastic
games, for general turn-based stochastic games the complexity of 
the algorithm of~\cite{GH07} is better.
However, for turn-based stochastic games where the transition function 
at all states can be expressed with constantly many bits we have 
$|\trans| =O(|S_R|)$.
In these cases the reachability strategy improvement algorithm (that works 
for both concurrent and turn-based stochastic games) 
works in time $2^{O(|S_R|)} \cdot \mathit{poly}(|G|)$ 
as compared to the time $2^{O(|S_R|\cdot \log(|S_R|)} \cdot \mathit{poly}(|G|)$
of  the algorithm of~\cite{GH07}.

\section{Existence of Memoryless Optimal Strategies for Concurrent Safety Games}
A proof of the existence of memoryless optimal strategies for
safety games can be found in \cite{dAM04}: 
the proof uses results on martingales to obtain the result.
For sake of completeness we present (an alternative) proof of the result:
the proof we present is similar in spirit with the other proofs in
this paper and uses the results on MDPs to obtain the result.
The proof is very similar to the proof presented in~\cite{EY06}.

\begin{theorem}{(Memoryless optimal strategies).}\label{thrm:memory-safe}
Memoryless optimal strategies exist for all
concurrent games with safety objectives.
\end{theorem}
\begin{proof}
Consider a concurrent game structure $G$ with an safety 
objective $\Safe(F)$ for player~1.
Then it follows from the results of~\cite{dAM04} that
\[
\va(\Safe(F)) = \nu X. \big(\min\set{[F], \Pre_1(X)} \big),
\] 
where $[F]$ is the indicator function of the set $F$ and 
$\nu$ denotes the greatest fixpoint.
Let $T=S \setminus F$, and for all states $s \in T$ we have  
$\va(\Safe(F))(s)=0$, and hence any memoryless strategy from $T$ is 
an optimal strategy.
Thus without loss of generality we assume all states in $T$ 
are absorbing. 
Let $v=\va(\Safe(F))$, and since we assume all states in $T$ are absorbing it 
follows that $\Pre_1(v)=v$ (since $v$ is a fixpoint).
Let $\gamma$ be a player~1 selector such that for all states $s$ we have 
$\Pre_{1:\gamma}(v)(s)=\Pre_1(v)(s)=v(s)$.
We show that $\ov{\gamma}$ is an memoryless optimal strategy.
Consider the player-2 MDP $G_{\gamma}$ and we consider the maximal 
probability for player~2 to reach the target set $T$.
Consider the valuation $w$ defined as $w=1-v$.
For all states $s \in T$ we have $w(s)=1$.
Since $\Pre_{1:\gamma}(v)=\Pre_1(v)$ it follows that 
for all states $s \in F$ and all $a_2 \in \mov_2(s)$ we have 
\[
\Pre_{\gamma,a_2}(v)(s) \geq \Pre_1(v)(s) = v(s);
\]
in other words, for all $s \in F$ we have 
$1-\Pre_1(v)(s) =1-v(s) \geq 1-\Pre_{\gamma,a_2}(v)(s)$.
Hence for all states $s\in F$ and all moves $a_2 \in \mov_2(s)$, we have
\[
  w(s) \geq \sum_{t \in S} w(t) \cdot \trans_{\gamma}(s,a_2). 
\]
Hence it follows that $w$ is a feasible solution to the 
linear program for MDPs with reachability objectives, i.e., 
given the memoryless strategy $\ov{\gamma}$ for player~1 
the maximal probability valuation for player~2 to reach $T$ 
is at most $w$.
Hence the memoryless strategy $\ov{\gamma}$ ensures that the 
probability valaution for player~1 to stay safe in $F$ against all 
player~2 strategies is at least $v=\va(\Safe(F))$.
Optimality of $\ov{\gamma}$ follows.
\qed
\end{proof}

\section{Strategy Improvement Algorithm for Concurrent Safety Games}

\noindent
In this section we present a strategy improvement 
algorithm for concurrent games with safety objectives. 
We consider a concurrent game structure with a safe set $F$, i.e.,
the objective for player~1 is $\Safe(F)$.
The algorithm will produce a sequence of selectors 
$\gamma_0, \gamma_1, \gamma_2, \ldots$ for player 1, such that
Condition~\ref{l-improve-1}, Condition~\ref{l-improve-3} and
Condition~\ref{l-improve-2} of Section~\ref{sec-stra-improve-reach} 
are satisfied.
Note that for concurrent safety games, there may be no $i \geq
0$ such that $\gamma_i = \gamma_{i+1}$, that is, the algorithm may fail to
generate an optimal selector, as the value can be irrational~\cite{dAM04}.
We start with a few notations

\medskip\noindent{\bf Optimal selectors.}
Given a valuation $v$ and a state $s$, we define by
\[
\OptSel(v,s) =\set{\xi_1 \in \Sel_1(s) \mid \Pre_{1:\xi_1}(v)(s) =\Pre_1(v)(s)}
\]
the set of optimal selectors for $v$ at state $s$.
For an optimal selector $\xi_1 \in \OptSel(v,s)$, we define the set
of counter-optimal actions as follows:
\[
\CountOpt(v,s,\xi_1) =\set{ b \in \mov_2(s) \mid 
\Pre_{\xi_1,b}(v)(s) =\Pre_1(v)(s)}.
\]
Observe that for $\xi_1 \in \OptSel(v,s)$, for all $b \in 
\mov_2(s) \setminus \CountOpt(v,s,\xi_1)$ we have 
$\Pre_{\xi_1,b}(v)(s) > \Pre_1(v)(s)$.
We define the set of optimal selector support and the counter-optimal 
action set as follows:
\[
\begin{array}{rcl}
\OptSelCount(v,s) & = & \set{(A,B) \subs \mov_1(s) \times \mov_2(s) \mid
\exists \xi_1 \in \Sel_1(s). \ \xi_1 \in \OptSel(v,s) \\
 &  & \land 
 \ \ \supp(\xi_1)=A \ \land \ \CountOpt(v,s,\xi_1)=B
};
\end{array}
\]
i.e., it consists of pairs $(A,B)$ of actions of player~1 and player~2,
such that there is an optimal selector $\xi_1$ with support $A$,
and $B$ is the set of counter-optimal actions to $\xi_1$.

\medskip\noindent{\bf Turn-based reduction.} Given a concurrent 
game $G=\langle S,\moves,\mov_1,\mov_2, \trans \rangle $ and 
a valuation $v$ we construct a turn-based stochastic game
$\ov{G}_v=\langle (\ov{S},\ov{E}), (\ov{S}_1,\ov{S}_2,\ov{S}_R),\ov{\trans}
\rangle$ as follows:
\begin{enumerate}
\item The set of states is as follows:
\[
\begin{array}{rcl}
\ov{S}& = & S \cup \set{(s,A,B) \mid s\in S, \ (A,B) \in \OptSelCount(v,s)} \\
	&\cup & \set{(s,A,b) \mid s \in S, \ (A,B) \in \OptSelCount(v,s), \ b \in B}.
\end{array}
\]

\item The state space partition is as follows: 
$\ov{S}_1=S$; $\ov{S}_2=\set{(s,A,B) \mid s \in S, (A,B) \in \OptSelCount(v,s)}$;
and $\ov{S}_R=\set{(s,A,b) \mid s\in S ,\  (A,B) \in \OptSelCount(v,s), b \in B}$.
In other words, $(\ov{S}_1,\ov{S}_2,\ov{S}_R)$ is a partition of the state 
space, where $\ov{S}_1$ are player~1 states, $\ov{S}_2$ are player~2 states,
and $\ov{S}_R$ are random or probabilistic states.

\item The set of edges is as follows:
\[ 
\begin{array}{rcl}
\ov{E} & = & \set{(s,(s,A,B)) \mid s \in S, (A,B) \in \OptSelCount(v,s)} \\
	& \cup & \set{((s,A,B),(s,A,b)) \mid b \in B} 
	\cup \set{((s,A,b),t) \mid \displaystyle t \in \bigcup_{a \in A} \dest(s,a,b)}.
\end{array}
\]

\item The transition function $\ov{\trans}$ for all states in $\ov{S}_R$ 
is uniform over its successors.
\end{enumerate}
Intuitively, the reduction is as follows.
Given the valuation $v$, state $s$ is a player~1 state where
player~1 can select a pair $(A,B)$ (and move to
state $(s,A,B)$) with $A \subs \mov_1(s)$ 
and $B \subs \mov_2(s)$ such that there is an optimal 
selector $\xi_1$ with support exactly $A$ and the set of
counter-optimal actions to $\xi_1$ is the set $B$.
From a player~2 state $(s,A,B)$, player~2 can choose any action
$b$ from the set $B$, and move to state $(s,A,b)$.
A state $(s,A,b)$ is a probabilistic state where all the 
states in $\bigcup_{a\in A} \dest(s,a,b)$ are chosen 
uniformly at random.
Given a set $F \subseteq S$ we denote by $\ov{F}= F \cup 
\set{(s,A,B) \in\ov{S} \mid s \in F} \cup 
\set{(s,A,b) \in\ov{S} \mid s \in F}$.
We refer to the above reduction as $\TB$, i.e., 
$(\ov{G}_v,\ov{F})=\TB(G,v,F)$.

\medskip\noindent{\bf Value-class of a valuation.}
Given a valuation $v$ and a real $0\leq r \leq 1$, 
the \emph{value-class} $U_r(v)$ of value $r$ is the set of 
states with valuation $r$, i.e., 
$U_r(v)=\set{s \in S \mid v(s)=r}$

\subsection{The strategy-improvement algorithm} 

\noindent{\bf Ordering of strategies.}
Let $G$ be a concurrent game and $F$ be the set of safe states.
Let $T =S \setminus F$.
Given a concurrent game structure $G$ with a safety objective $\Safe(F)$,
the set of \emph{almost-sure winning} states is the set of states $s$ such that
the value at $s$ is~$1$, i.e., $W_1=\set{s\in S \mid \va(\Safe(F))=1}$
is the set of almost-sure winning states.
An optimal strategy from $W_1$ is referred as an almost-sure winning 
strategy.
The set $W_1$ and an almost-sure winning strategy can be computed in 
linear time by the algorithm given in~\cite{dAH00}.
We assume without loss of generality that all states in $W_1 \union T$
are absorbing. 
We recall the preorder $\prec$ on the strategies for player 1 (as defined in 
Section~\ref{subsec:reach-stra})
as follows:
given two player 1 strategies $\stra_1$ and $\stra_1'$, let
$\stra_1 \prec \stra_1'$ if the following two conditions hold:
(i)~$\vas{\stra_1}(\Safe(F)) \leq \vas{\stra_1'}(\Safe(F))$; and 
(ii)~$\vas{\stra_1}(\Safe(F))(s) < \vas{\stra_1'}(\Safe(F))(s)$ 
for some state $s\in S$.
Furthermore, we write 
$\stra_1 \preceq \stra_1'$ if either $\stra_1 \prec \stra_1'$ or 
$\stra_1 = \stra_1'$.
We first present an example that shows the improvements 
based only on $\Pre_1$ operators are not sufficient for 
safety games, even on turn-based games and then present our algorithm.

\begin{example}{}\label{examp:conc-safety}
Consider the turn-based stochastic game shown in Fig~\ref{fig:example-tbs}, where
the $\Box$ states are player~1 states, the $\Diamond$ states are 
player~2 states, and $\bigcirc$ states are random states with probabilities
labeled on edges.
The safety goal is to avoid the state $s_4$. 
Consider a memoryless strategy $\stra_1$ for player~1 that chooses the successor $s_0 \to
s_2$, and the counter-strategy $\stra_2$ for player~2 chooses $s_1 \to s_0$.
Given the strategies $\stra_1$ and $\stra_2$, the value at 
$s_0,s_1$ and $s_2$ is $1/3$, and since all
successors of $s_0$ have value $1/3$, the value cannot be improved by $\Pre_1$.
However, note that if player~2 is restricted to choose only value optimal 
selectors for the value $1/3$, then player~1 can switch to the strategy 
$s_0\to s_1$ and ensure that the game stays in the value class $1/3$ 
with probability~1.
Hence switching to $s_0 \to s_1$ would force player~2 to select a 
counter-strategy that switches to the strategy $s_1 \to s_3$, and
thus player~1 can get a value $2/3$.
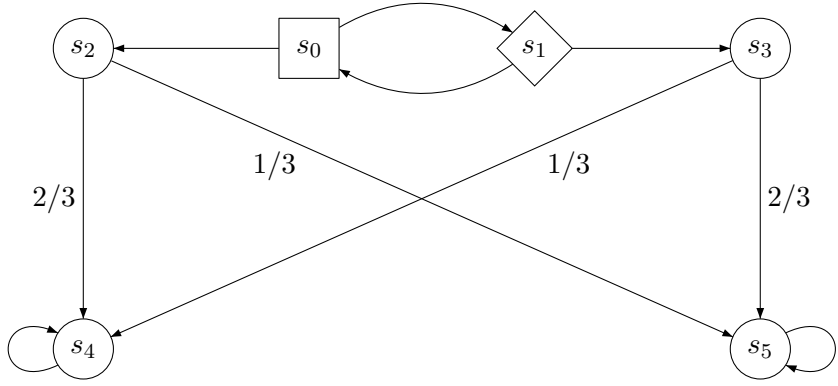
\begin{figure}[!tb]
\begin{center} 
\def\fsize{\normalsize}

\begin{picture}(75,45)(15,0)
{\fsize
\node[Nmarks=n](n2)(10,50){$s_2$}
\node[Nmarks=n](n4)(10,10){$s_4$}
\node[Nmarks=n](n3)(100,50){$s_3$}
\node[Nmarks=n](n5)(100,10){$s_5$}
\node[Nmarks=n, Nmr=0](n0)(40,50){$s_0$}
\rpnode[Nmarks=n](n1)(70,50)(4,5.0){$s_1$}
\drawloop[loopangle=180, loopdiam=6](n4){}
\drawloop[loopangle=0, loopdiam=6](n5){}
\drawedge[ELpos=50, ELside=l, curvedepth=0](n0,n2){}
\drawedge[ELpos=50, ELside=l, curvedepth=0](n1,n3){}
\drawedge[ELpos=50, ELside=r, curvedepth=0](n2,n4){$2/3$}
\drawedge[ELpos=30, ELside=r, curvedepth=0](n2,n5){$1/3$}
\drawedge[ELpos=50, ELside=l, curvedepth=0](n3,n5){$2/3$}
\drawedge[ELpos=30, ELside=l, curvedepth=0](n3,n4){$1/3$}
\drawedge[ELpos=50, ELside=l, curvedepth=6](n0,n1){}
\drawedge[ELpos=30, ELside=l, curvedepth=6](n1,n0){}
}
\end{picture}
\end{center}
\caption{A turn-based stochastic safety game.}\label{fig:example-tbs}
\end{figure}
\qed
\end{example}

\smallskip\noindent{\bf Informal description of Algorithm~\ref{algorithm:strategy-improve-safe}.}
We first present the basic strategy improvement algorithm 
(Algorithm~\ref{algorithm:strategy-improve-safe}) and will later 
present a convergent version (Algorithm~\ref{algorithm:strategy-convergent}) 
for computing the values for all states in $S \setminus W_1$.
The algorithm (Algorithm~\ref{algorithm:strategy-improve-safe}) iteratively 
improves player-1 strategies according to the preorder $\prec$. 
The algorithm starts with the random selector 
$\gamma_0=\overline{\xi}_1^\unif$ that plays at all states all actions
uniformly at random.
At iteration $i+1$, the algorithm considers the memoryless player-1 strategy
$\overline{\gamma}_i$ and computes the value $\vas{\overline{\gamma}_i}(\Safe(F))$.
Observe that since $\overline{\gamma}_i$ is a memoryless strategy, the
computation of $\vas{\overline{\gamma}_i}(\Safe(F))$ involves solving the 2-MDP 
$G_{{\gamma}_i}$.
The valuation $\vas{\overline{\gamma}_i}(\Safe(F))$ is named $v_i$.
For all states $s$ such that $\Pre_1(v_i)(s) > v_i(s)$, 
the memoryless strategy at $s$ is modified to a selector 
that is value-optimal for $v_i$. 
The algorithm then proceeds to the next iteration.
If $\Pre_1(v_i) = v_i$, then the algorithm constructs the 
game $(\ov{G}_{v_i},\ov{F})=\TB(G,v_i,F)$, and computes
$\ov{A}_i$ as the set of almost-sure winning states in $\ov{G}_{v_i}$
for the objective $\Safe(\ov{F})$.
Let $U=(\ov{A}_i \cap S) \setminus W_1$.
If $U$ is non-empty, then a selector $\gamma_{i+1}$ is obtained at $U$ 
from an pure memoryless optimal strategy (i.e.,
an almost-sure winning strategy) in $\ov{G}_{v_i}$, and
the algorithm proceeds to iteration $i+1$.
If $\Pre_1(v_i)=v_i$ and $U$ is empty, then the algorithm stops and returns 
the memoryless strategy $\overline{\gamma}_i$ for player~1.
Unlike strategy improvement algorithms for turn-based games (see
\cite{Con93} for a survey),
Algorithm~\ref{algorithm:strategy-improve-safe} is not guaranteed to
terminate (see Example~\ref{ex:counter-soda}).
We will show that Algorithm~\ref{algorithm:strategy-improve-safe} 
has both the monotonicity and optimality on termination properties,
however, as we will illustrate in Example~\ref{ex:counter-soda}, 
the valuations of Algorithm~\ref{algorithm:strategy-improve-safe}
need not necessarily converge to the values.
However, for turn-based stochastic games Algorithm~\ref{algorithm:strategy-improve-safe}
correctly converges to the values.
We will show that Algorithm~\ref{algorithm:strategy-convergent} has 
all the desired properties (i.e., monotonicity, optimality on termination,
and convergence to the values).

\begin{algorithm*}[t]
\caption{Safety Strategy-Improvement Algorithm}
\label{algorithm:strategy-improve-safe}
{
\begin{tabbing}
aaa \= aaa \= aaa \= aaa \= aaa \= aaa \= aaa \= aaa \kill
\\
\> {\bf Input:} a concurrent game structure $G$ with safe set $F$. \\
\>   {\bf Output:} a strategy $\overline{\gamma}$ for player~1. \\ 

\> 0. Compute $W_1=\set{s \in S \mid \va(\Safe(F))(s)=1}$. \\
\> 1. Let $\gamma_0=\xi_1^\unif$ and $i=0$. \\
\> 2. Compute $v_0 = \vas{\overline{\gamma}_0}(\Safe(F))$. \\

\> 3. {\bf do \{ } \\ 
\>\> 3.1. Let $I= \set{s \in S \setminus (W_1 \cup T) \mid \Pre_1(v_i)(s) > v_i(s)}$. \\
\>\> 3.2 {\bf if} $I \neq \emptyset$, {\bf then} \\
\>\>\> 3.2.1 Let $\xi_1$ be a player-1 
        selector such that for all states $s \in I$, \\
\>\>\>\>	we have $\Pre_{1:\xi_1}(v_i)(s) =\Pre_1(v_i)(s) > v_i(s)$.\\
\>\>\> 3.2.2 The player-1 selector $\gamma_{i+1}$ is defined
          as follows: for each state $s\in S$, let\\
\>\>\>\> $	\displaystyle 
	\gamma_{i+1}(s)=
	\begin{cases}
	\gamma_i(s) & \text{\ if \ }s\not\in I;\\
	\xi_1(s) & \text{\ if\ }s\in I.
	\end{cases}$
	\\ 
\>\> 3.3 {\bf else}  \\
\>\>\> 3.3.1 let$(\ov{G}_{v_i},\ov{F})=\TB(G,v_i,F)$ \\
\>\>\> 3.3.2 let $\ov{A}_i$ be the set of almost-sure winning states in $\ov{G}_{v_i}$
	for $\Safe(\ov{F})$ and \\
\>\>\>\> $\ov{\stra}_1$ be a pure memoryless almost-sure winning strategy from the set $\ov{A}_i$.\\
\>\>\> 3.3.3 {\bf if} ($(\ov{A}_i \cap S) \setminus W_1 \neq \emptyset$) \\
\>\>\>\> 3.3.3.1 let $U= (\ov{A}_i \cap S)\setminus W_1$ \\
\>\>\>\> 3.3.3.2 The player-1 selector $\gamma_{i+1}$ is defined
          as follows: for $s\in S$, let\\
\>\>\>\> $	\displaystyle 
	\gamma_{i+1}(s)=
	\begin{cases}
	\gamma_i(s) & \text{\ if \ }s\not\in U;\\
	\xi_1(s) & \text{\ if\ }s\in U, \xi_1(s) \in \OptSel(v_i,s), 
				\supp(\xi_1(s))=A, \\
	&\ \ \ov{\stra}_1(s)=(s,A,B), B=\CountOpt(s,v,\xi_1).
	\end{cases}$
	\\
\>\> 3.4. Compute $v_{i+1} =\vas{\overline{\gamma}_{i+1}}(\Safe(F))$. \\
\>\> 3.5. Let $i=i+1$. \\
\> {\bf \} until } $I=\emptyset$ and $(\ov{A}_{i-1} \cap S) \setminus W_1=\emptyset$. \\
\> 4. {\bf return} $\overline{\gamma}_{i}$.  
\end{tabbing}
}
\end{algorithm*}

\begin{algorithm*}[t]
\caption{$k$-Uniform Restricted Safety Strategy-Improvement Algorithm}
\label{algo:k-uniform}
{
\begin{tabbing}
aaa \= aaa \= aaa \= aaa \= aaa \= aaa \= aaa \= aaa \kill
\\
\> {\bf Input:} a concurrent game structure $G$ with safe set $F$, and number $k$. \\
\>   {\bf Output:} a strategy $\overline{\gamma}$ for player~1. \\ 

\> 0. Compute $W_1=\set{s \in S \mid \va(\Safe(F))(s)=1}$; and $k=\max\set{k,|\moves|}$. \\
\> 1. Let $\gamma_0=\xi_1^\unif$ and $i=0$. \\
\> 2. Compute $v_0 = \vas{\overline{\gamma}_0}(\Safe(F))$. \\

\> 3. {\bf do \{ } \\ 
\>\> 3.1. Let $I_k= \set{s \in S \setminus (W_1 \cup T) \mid \sup_{\xi_1' \in \Sel^k(s)} \Pre_{1:\xi_1'}(v_i)(s) > v_i(s)}$. \\
\>\> 3.2 {\bf if} $I_k \neq \emptyset$, {\bf then} \\
\>\>\> 3.2.1 Let $\xi_1$ be a $k$-uniform player-1 
        selector such that for all states $s \in I$, \\
\>\>\>\>	we have $\Pre_{1:\xi_1}(v_i)(s) = \sup_{\xi_1' \in \Sel^k(s)} \Pre_{1:\xi_1'}(v_i)(s) > v_i(s)$.\\
\>\>\> 3.2.2 The player-1 selector $\gamma_{i+1}$ is defined
          as follows: for each state $s\in S$, let\\
\>\>\>\> $	\displaystyle 
	\gamma_{i+1}(s)=
	\begin{cases}
	\gamma_i(s) & \text{\ if \ }s\not\in I_k;\\
	\xi_1(s) & \text{\ if\ }s\in I_k.
	\end{cases}$
	\\ 
\>\> 3.3 {\bf else}  \\
\>\>\> 3.3.1 let$(\ov{G}_{v_i}^k,\ov{F})=\TB(G,v_i,F,k)$ \\
\>\>\> 3.3.2 let $\ov{A}_i^k$ be the set of almost-sure winning states in $\ov{G}_{v_i}^k$
	for $\Safe(\ov{F})$ and \\
\>\>\>\> $\ov{\stra}_1$ be a pure memoryless almost-sure winning strategy from the set $\ov{A}_i^k$.\\
\>\>\> 3.3.3 {\bf if} ($(\ov{A}_i^k \cap S) \setminus W_1 \neq \emptyset$) \\
\>\>\>\> 3.3.3.1 let $U= (\ov{A}_i^k \cap S)\setminus W_1$ \\
\>\>\>\> 3.3.3.2 The player-1 selector $\gamma_{i+1}$ is defined
          as follows: for $s\in S$, let\\
\>\>\>\> $	\displaystyle 
	\gamma_{i+1}(s)=
	\begin{cases}
	\gamma_i(s) & \text{\ if \ }s\not\in U;\\
	\xi_1(s) & \text{\ if\ }s\in U, \xi_1(s) \in \OptSel(v_i,s,k), \supp(\xi_1(s))=A, \\
	&\ \ \ov{\stra}_1(s)=(s,A,B), B=\CountOpt(s,v,\xi_1,k).
	\end{cases}$
	\\
\>\> 3.4. Compute $v_{i+1} =\vas{\overline{\gamma}_{i+1}}(\Safe(F))$. \\
\>\> 3.5. Let $i=i+1$. \\
\> {\bf \} until } $I_k=\emptyset$ and $(\ov{A}_{i-1}^k \cap S) \setminus W_1=\emptyset$. \\
\> 4. {\bf return} $\overline{\gamma}_{i}$.  
\end{tabbing}
}
\end{algorithm*}

\begin{lemma}{}\label{lemm:stra-improve-safe1}
Let $\gamma_i$ and $\gamma_{i+1}$ be the player-1 selectors obtained at 
iterations $i$ and $i+1$ of Algorithm~\ref{algorithm:strategy-improve-safe}.
Let $I=\set{s \in S \setminus (W_1 \cup T) \mid \Pre_1(v_i)(s) > v_i(s)}$. 
Let $v_i=\vas{\overline{\gamma}_i}(\Safe(F))$ and 
$v_{i+1}=\vas{\overline{\gamma}_{i+1}}(\Safe(F))$.
Then $v_{i+1}(s)  \geq  \Pre_1(v_i)(s)$ for all states $s\in S$;
and therefore
$v_{i+1}(s) \geq v_i(s)$ for all states $s\in S$, 
and $v_{i+1}(s) > v_i(s)$ for all states $s\in I$.
\end{lemma}
\begin{proof}
The proof is essentially similar to the proof of Lemma~\ref{lemm:stra-improve}, and
we present the details for completeness.
Consider the valuations $v_i$ and $v_{i+1}$ obtained at iterations $i$ and 
$i+1$, respectively, and let $w_i$ be the valuation defined by 
$w_i(s) = 1 - v_i(s)$ for all states $s \in S$. 
The counter-optimal strategy for player~2 to minimize
$v_{i+1}$ is obtained by maximizing the probability to reach $T$.
Let 
\[
  \wh{w}_{i}(s) =
  \begin{cases} 
    w_i(s) & \text{\ if\ }s \in S \setminus I; \\
    1-\Pre_1(v_i)(s) < w_i(s) &\text{\ if\ }s \in I. 
  \end{cases}
\]
In other words, $\wh{w}_{i} = 1 - \Pre_1(v_i)$, and we also have 
$\wh{w}_{i} \leq w_i$.
We now show that $\wh{w}_{i}$ is a feasible solution to the 
linear program for MDPs with the objective $\Reach(T)$, as
described in Section~\ref{sec:mdp}.
Since $v_i =\vas{\overline{\gamma}_i}(\Safe(F))$, it follows that
for all states $s\in S$ and all moves $a_2 \in \mov_2(s)$, we have
\[
  w_i(s) \geq \sum_{t \in S} w_i(t) \cdot \trans_{\gamma_i}(s,a_2).  
\]
For all states $s \in S \setminus I$, 
we have $\gamma_i(s)=\gamma_{i+1}(s)$ and
$\wh{w}_{i}(s) = w_i(s)$,
and since $\wh{w}_{i} \leq w_i$, it follows that for all 
states $s \in S \setminus I$ and all moves $a_2 \in \mov_2(s)$, we have
\[
  \wh{w}_{i}(s)=w_i(s) \geq \sum_{t \in S} \wh{w}_{i}(t) \cdot \trans_{\gamma_{i+1}}(s,a_2) 
\qquad \text{( for  $s\in S\setm I$)}.
\]

Since for $s \in I$ the selector $\gamma_{i+1}(s)$ is obtained as an
optimal selector for $\Pre_1(v_i)(s)$, it follows that 
for all states $s\in I$ and all moves $a_2 \in \mov_2(s)$, we have
\[
\Pre_{\gamma_{i+1},a_2}(v_i)(s) \geq \Pre_1(v_i)(s);
\]
in other words, $1-\Pre_1(v_i)(s) \geq 1-\Pre_{\gamma_{i+1},a_2}(v_i)(s)$.
Hence for all states $s\in I$ and all moves $a_2 \in \mov_2(s)$, we have
\[
  \wh{w}_{i}(s) \geq \sum_{t \in S} w_{i}(t) \cdot \trans_{\gamma_{i+1}}(s,a_2). 
\]
Since $\wh{w}_{i} \leq w_i$,  
for all states $s\in I$ and all moves $a_2 \in \mov_2(s)$, we have
\[
  \wh{w}_{i}(s) \geq \sum_{t \in S} \wh{w}_{i}(t) \cdot \trans_{\gamma_{i+1}}(s,a_2) 
\qquad \text{( for  $s\in I$)}.
\] 
Hence it follows that $\wh{w}_{i}$ is a feasible solution to the 
linear program for MDPs with reachability objectives.
Since the reachability valuation 
for player~2 for $\Reach(T)$ is the least
solution (observe that the objective function of the linear program is a 
minimizing function), it follows that 
$v_{i+1} \geq 1-\wh{w}_{i} = \Pre_1(v_i)$.
Thus we obtain
$v_{i+1}(s) \geq v_i(s)$ for all states $s \in S$, and 
$v_{i+1}(s) > v_i(s)$ for all states $s \in I$.
\qed
\end{proof}

Recall that by Example~\ref{examp:conc-safety} it follows that 
improvement by only step~3.2 is not sufficient to guarantee
convergence to optimal values.
We now present a lemma about the turn-based reduction,
and then show that step 3.3 also leads to an improvement.
Finally, in Theorem~\ref{thrm:safe-termination} we show that if
improvements by step 3.2 and step 3.3 are not possible, then 
the optimal value and an optimal strategy is obtained.

\begin{lemma}{}\label{lemm:stra-improve-safetb}
Let $G$ be a concurrent game with a set $F$ of safe states.
Let $v$ be a valuation and 
consider $(\ov{G}_v,\ov{F})=\TB(G,v,F)$.
Let $\ov{A}$ be the set of almost-sure winning states in $\ov{G}_v$
for the objective $\Safe(\ov{F})$, and let $\ov{\stra}_1$ be a
pure memoryless almost-sure winning strategy from $\ov{A}$ in
$\ov{G}_v$.
Consider a memoryless strategy $\stra_1$ in $G$ for 
states in $\ov{A}\cap S$ as follows:
if $\ov{\stra}_1(s)=(s,A,B)$, then $\stra_1(s) \in 
\OptSel(v,s)$ such that $\supp(\stra_1(s))=A$ and $\CountOpt(v,s,\stra_1(s))=B$.
Consider a pure memoryless strategy $\stra_2$ 
for player~2.
If for all states $s \in \ov{A} \cap S$, we have 
$\stra_2(s) \in \CountOpt(v,s,\stra_1(s))$, then 
for all $s \in \ov{A} \cap S$, we have 
$\Prb_s^{\stra_1,\stra_2}(\Safe(F))=1$. 
\end{lemma}
\begin{proof}
We analyze the Markov chain arising after the player fixes
the memoryless strategies $\stra_1$ and $\stra_2$.
Given the strategy $\stra_2$ consider the strategy $\ov{\stra}_2$
as follows: if $\ov{\stra}_1(s)=(s,A,B)$ and $\stra_2(s)=b \in 
\CountOpt(v,s,\stra_1(s))$, then at state $(s,A,B)$ choose
the successor $(s,A,b)$. 
Since $\ov{\stra}_1$ is an almost-sure winning strategy for 
$\Safe(\ov{F})$, it follows that in the Markov chain obtained by 
fixing $\ov{\stra}_1$ and $\ov{\stra}_2$ in $\ov{G}_v$, 
all closed connected recurrent set of states that intersect 
with $\ov{A}$ are contained in $\ov{A}$, and from all 
states of $\ov{A}$ the closed connected recurrent set of states 
within $\ov{A}$ are reached with probability~1. 
It follows that in the Markov chain obtained from fixing 
$\stra_1$ and $\stra_2$ in $G$ 
all closed connected recurrent set of states that intersect 
with $\ov{A}\cap S $ are contained in $\ov{A} \cap S$, and from all 
states of $\ov{A}\cap S$ the closed connected recurrent set of states 
within $\ov{A} \cap S$ are reached with probability~1. 
The desired result follows. 
\qed
\end{proof}

\begin{lemma}{}\label{lemm:stra-improve-safe2}
Let $\gamma_i$ and $\gamma_{i+1}$ be the player-1 selectors obtained at 
iterations $i$ and $i+1$ of Algorithm~\ref{algorithm:strategy-improve-safe}.
Let $I=\set{s \in S \setminus (W_1 \cup T) \mid \Pre_1(v_i)(s) > v_i(s)}=\emptyset$,
and $(\ov{A}_i \cap S)\setminus W_1 \neq \emptyset$.  
Let $v_i=\vas{\overline{\gamma}_i}(\Safe(F))$ and 
$v_{i+1}=\vas{\overline{\gamma}_{i+1}}(\Safe(F))$.
Then 
$v_{i+1}(s) \geq v_i(s)$ for all states $s\in S$, 
and $v_{i+1}(s) > v_i(s)$ for some state $s\in (\ov{A}_i \cap S) \setminus W_1$.
\end{lemma}
\begin{proof}
We first show that $v_{i+1} \geq v_i$.
Let $U=(\ov{A}_i \cap S)\setminus W_1$.
Let $w_i(s) = 1 - v_i(s)$ for all states $s \in S$. 
Since $v_i =\vas{\overline{\gamma}_i}(\Safe(F))$, it follows that
for all states $s\in S$ and all moves $a_2 \in \mov_2(s)$, we have
\[
  w_i(s) \geq \sum_{t \in S} w_i(t) \cdot \trans_{\gamma_i}(s,a_2).  
\]
The selector $\xi_1(s)$ chosen for $\gamma_{i+1}$ at $s \in U$ satisfies that
$\xi_1(s) \in \OptSel(v_i,s)$.
It follows that for all states $s\in S$ and all moves $a_2 \in \mov_2(s)$, 
we have
\[
  w_i(s) \geq \sum_{t \in S} w_i(t) \cdot \trans_{\gamma_{i+1}}(s,a_2).  
\]
It follows that the maximal probability with which player~2 can reach
$T$ against the strategy $\ov{\gamma}_{i+1}$ is at most $w_i$.
It follows that $v_i(s) \leq v_{i+1}(s)$.

We now argue that for some state $s \in U$ we have $v_{i+1}(s)>v_i(s)$.
Given the strategy $\ov{\gamma}_{i+1}$, consider a pure memoryless
counter-optimal strategy $\stra_2$ for player~2 to reach $T$.
Since the selectors $\gamma_{i+1}(s)$ at states $s\in U$ are obtained from the 
almost-sure strategy $\ov{\stra}$ in the turn-based game $\ov{G}_{v_i}$ 
to satisfy $\Safe(\ov{F})$, 
it follows from Lemma~\ref{lemm:stra-improve-safetb} 
that if for every state $s \in U$, the action 
$\stra_2(s) \in \CountOpt(v_i,s,\gamma_{i+1})$, then from 
all states $s \in U$, the game stays safe in $F$ with probability~1.
Since $\ov{\gamma}_{i+1}$ is a given strategy for player~1, and 
$\stra_2$ is counter-optimal against $\ov{\gamma}_{i+1}$, this 
would imply that $U\subseteq \set{s \in S \mid \va(\Safe(F))=1}$.
This would contradict that $W_1=\set{s\in S \mid \va(\Safe(F))=1}$ 
and $U \cap W_1=\emptyset$.
It follows that for some state $s^* \in U$ we have $\stra_2(s^*)
\not\in \CountOpt(v_i,s^*,\gamma_{i+1})$,
and since $\gamma_{i+1}(s^*) \in \OptSel(v_i,s^*)$ 
we have 
\[
v_i(s^*) < \sum_{t \in S} v_i(t) \cdot \trans_{\gamma_{i+1}}(s^*,\stra_2(s^*));
\]
in other words, we have
\[
w_i(s^*) > \sum_{t \in S} w_i(t) \cdot \trans_{\gamma_{i+1}}(s^*,\stra_2(s^*)).
\]
Define a valuation $z$ as follows:
$z(s)=w_i(s)$ for $s \neq s^*$, and 
$z(s^*)=\sum_{t \in S} w_i(t) \cdot \trans_{\gamma_{i+1}}(s^*,\stra_2(s^*))$.
Given the strategy $\ov{\gamma}_{i+1}$ and the
counter-optimal strategy $\stra_2$, the valuation $z$ satisfies the
inequalities of the linear-program for reachability to $T$.
It follows that the probability to reach $T$ given $\ov{\gamma}_{i+1}$ 
is at most $z$.
Thus we obtain
that $v_{i+1}(s) \geq v_i(s)$ for all $s\in S$,
and $v_{i+1}(s^*) > v_i(s^*)$.
This concludes the proof.
\qed
\end{proof}

We obtain the following theorem from Lemma~\ref{lemm:stra-improve-safe1}
and Lemma~\ref{lemm:stra-improve-safe2} that shows that the sequences of
values we obtain is monotonically non-decreasing.

\begin{theorem}{(Monotonicity of values).}\label{thrm:safe-mono}
For $i\geq 0$, let $\gamma_{i}$ and $\gamma_{i+1}$ be the player-1 selectors obtained
at iterations $i$ and $i+1$ of Algorithm~\ref{algorithm:strategy-improve-safe}.
If $\gamma_{i}\neq \gamma_{i+1}$, then 
(a)~for all $s \in S$ we have 
$\vas{\overline{\gamma}_i}(\Safe(F))(s) \leq \vas{\overline{\gamma}_{i+1}}(\Safe(F))(s)$; and
(b)~for some $s^* \in S$ we have 
$\vas{\overline{\gamma}_i}(\Safe(F))(s^*) < \vas{\overline{\gamma}_{i+1}}(\Safe(F))(s^*)$.
\end{theorem}

\begin{theorem}{(Optimality on termination).}\label{thrm:safe-termination}
Let $v_i$ be the valuation at iteration $i$ of 
Algorithm~\ref{algorithm:strategy-improve-safe} such that 
$v_i=\vas{\ov{\gamma}_i}(\Safe(F))$. 
If  
$I=\set{s\in S \setminus (W_1 \cup T) \mid \Pre_1(v_i)(s) > v_i(s)}=\emptyset$,
and $(\ov{A}_i \cap S)\setminus W_1=\emptyset$,
then $\ov{\gamma}_i$ is an optimal strategy and 
$v_i=\va(\Safe(F))$.
\end{theorem}
\begin{proof}
We show that for all memoryless strategies $\stra_1$ for player~1 we have 
$\vas{\stra_1}(\Safe(F)) \leq v_i$.
Since memoryless optimal strategies exist for concurrent games with safety objectives
(Theorem~\ref{thrm:memory-safe}) the desired result follows.

Let $\ov{\stra}_2$ be a pure memoryless optimal strategy for 
player~2 in $\ov{G}_{v_i}$ for the objective 
complementary to $\Safe(\ov{F})$, 
where $(\ov{G}_{v_i},\Safe(\ov{F}))=\TB(G,v_i,F)$.
Consider a memoryless strategy $\stra_1$ for player~1,
and we define a pure memoryless strategy $\stra_2$
for player~2 as follows.
\begin{enumerate}
\item If $\stra_1(s) \not\in \OptSel(v_i,s)$, then $\stra_2(s)=b \in \mov_2(s)$,
	such that $\Pre_{\stra_1(s),b}(v_i)(s) < v_i(s)$;
	(such a $b$ exists since $\stra_1(s) \not\in \OptSel(v_i,s)$).

\item If $\stra_1(s) \in \OptSel(v_i,s)$, then let $A=\supp(\stra_1(s))$,
	and consider $B$ such that $B=\CountOpt(v_i,s,\stra_1(s))$.
	Then we have $\stra_2(s)=b$, such that $\ov{\stra}_2((s,A,B))=(s,A,b)$. 
\end{enumerate}
Observe that by construction of $\stra_2$, for all 
$s \in S \setminus (W_1 \cup T)$, we have 
$\Pre_{\stra_1(s),\stra_2(s)}(v_i)(s) \leq v_i(s)$.
We first show that in the Markov chain obtained by fixing $\stra_1$ and 
$\stra_2$ in $G$, there is no closed connected recurrent set of states $C$
such that $C \subseteq S \setminus (W_1 \cup T)$.
Assume towards contradiction that $C$ is a closed connected recurrent 
set of states in $S \setminus (W_1 \cup T)$.
The following case analysis achieves the contradiction.
\begin{enumerate}
\item Suppose for every state $s \in C$ we have $\stra_1(s) \in \OptSel(v_i,s)$.
Then consider the strategy $\ov{\stra}_1$ in $\ov{G}_{v_i}$ such that 
for a state $s \in C$ we have $\ov{\stra}_1(s)=(s,A,B)$,
where $\stra_1(s)=A$, and $B=\CountOpt(v_i,s,\stra_1(s))$.
Since $C$ is closed connected recurrent states, it follows by construction 
that for all states $s \in C$ in the game $\ov{G}_{v_i}$ we have 
$\Prb_s^{\ov{\stra}_1,\ov{\stra}_2}(\Safe(\ov{C}))=1$,
where $\ov{C}=C \cup \set{(s,A,B) \mid s \in C} \cup \set{(s,A,b) \mid s \in C}$.
It follows that for all $s \in C$ in $\ov{G}_{v_i}$ 
we have $\Prb_s^{\ov{\stra}_1,\ov{\stra}_2}(\Safe(\ov{F}))=1$.
Since $\ov{\stra}_2$ is an optimal strategy, it follows that $C 
\subseteq (\ov{A}_i \cap S)\setminus W_1$.
This contradicts that $(\ov{A}_i \cap S) \setminus W_1=\emptyset$.

\item Otherwise for some state $s^* \in C$ we have $\stra_1(s^*) \not\in
\OptSel(v_i,s^*)$.
Let $r=\min\set{q \mid U_q(v_i) \cap C \neq \emptyset}$, i.e., 
$r$ is the least value-class with non-empty intersection with $C$.
Hence it follows that for all $q<r$, we have 
$U_q(v_i) \cap C=\emptyset$.
Observe that since for all $s \in C$ we have 
$\Pre_{\stra_1(s),\stra_2(s)}(v_i)(s) \leq v_i(s)$,
it follows that for all $s \in U_r(v_i)$ either
(a)~$\dest(s,\stra_1(s),\stra_2(s))\subseteq U_r(v_i)$;
or (b)~$\dest(s,\stra_1(s),\stra_2(s)) \cap U_q(v_i) \neq \emptyset$,
for some $q<r$.
Since $U_r(v_i)$ is the least value-class with non-empty intersection 
with $C$, it follows that for all $s \in U_r(v_i)$ we have 
$\dest(s,\stra_1(s),\stra_2(s)) \subseteq U_r(v_i)$.
It follows that $C \subseteq U_r(v_i)$. 
Consider the state $s^* \in C$ such that $\stra_1(s^*) \not\in \OptSel(v_i,s)$.
By the construction of $\stra_2(s)$, we have 
$\Pre_{\stra_1(s^*),\stra_2(s^*)}(v_i)(s^*) < v_i(s^*)$.
Hence we must have $\dest(s^*,\stra_1(s^*),\stra_2(s^*)) \cap U_{q}(v_i) 
\neq \emptyset$, for some $q <r$.
Thus we have a contradiction.
\end{enumerate}
It follows from above that there is no closed connected recurrent set of states
in $S\setminus (W_1 \cup T)$, and hence with probability~1 
the game reaches $W_1 \cup T$ from all states in $S \setminus (W_1 \cup T)$.
Hence the probability to satisfy $\Safe(F)$ is equal to the probability 
to reach $W_1$.
Since for all states $s \in S\setminus (W_1 \cup T)$ we have 
$\Pre_{\stra_1(s),\stra_2(s)}(v_i)(s) \leq v_i(s)$, 
it follows that given the strategies $\stra_1$ and 
$\stra_2$, the valuation $v_i$ satisfies all the inequalities 
for linear program to reach $W_1$.
It follows that the probability to reach $W_1$ from $s$ is 
atmost $v_i(s)$.
It follows that for all $s \in S\setminus (W_1 \cup T)$ 
we have $\vas{\stra_1}(\Safe(F))(s)\leq v_i(s)$.
The result follows.
\qed
\end{proof}

\medskip\noindent{\bf $k$-uniform selectors and strategies.}
For concurrent games, we will use the result that for 
$\vare>0$, there is a \emph{$k$-uniform memoryless} strategy
that achieves the value of a safety objective within $\vare$.
We first define $k$-uniform selectors and $k$-uniform 
memoryless strategies.
For a positive integer $k>0$, a selector $\xi$ for player~1 is 
\emph{$k$-uniform} if for all $s \in S \setminus (T \cup W_1)$ and all 
$a \in \supp(\stra_1(s))$ there exists $i,j \in \Nats$ such that 
$0 \leq i \leq j \leq k$ and $\xi(s)(a)=\frac{i}{j}$, i.e., the moves in the 
support are  played with probability that are multiples of 
$\frac{1}{\ell}$ with $\ell \leq k$.
We denote by $\Sel^k$ the set of $k$-uniform selectors. 
A memoryless strategy is $k$-uniform if it is obtained from a $k$-uniform 
selector.
We denote by $\Stra_1^{M,k}$ the set of $k$-uniform memoryless strategies
for player~1.
We first present a technical lemma (Lemma~\ref{lemm-expl-constr}) that will
be used in the key lemma (Lemma~\ref{lemm:kuniform}) to prove the 
convergence result.

\begin{lemma}{}\label{lemm-expl-constr}
Let $a_1,a_2,\ldots,a_m$ be $m$ real numbers such that 
(1) for all $1 \leq i \leq m$, we have $a_i>0$; and
(2) $\sum_{i=1}^m a_i=1$.
Let $c=\min_{1\leq i\leq m} a_i$.
For $\eta>0$, there exists $k \geq \frac{m}{c\cdot \eta}$ and 
$m$ real numbers $b_1, b_2, \ldots,b_m$ such that 
(1) for all $1 \leq i \leq m$, we have $b_i$ is a multiple of $\frac{1}{k}$ and $b_i>0$; 
(2) $\sum_{i=1}^m b_i=1$; and 
(3) for all $1 \leq i \leq m$, we have $\frac{a_i}{b_i} \leq 1 + \eta$ and 
$\frac{b_i}{a_i} \leq 1 + \eta$. 
\end{lemma}
\begin{proof}
Let $\ell=\frac{m}{\eta \cdot c}$.
For $1 \leq i \leq m$, define $\ov{b}_i$ such that $\ov{b}_i$ is a multiple of 
$\frac{1}{\ell}$ and $a_i \leq \ov{b}_i \leq a_i + \frac{1}{\ell}$ 
(basically define $\ov{b}_i$ as the least multiple of $\frac{1}{\ell}$ that is
at least the value of $a_i$).
For $1 \leq i \leq m$, let $b_i= \frac{\ov{b}_i}{ \sum_{i=1}^m \ov{b}_i}$; 
i.e., $b_i$ is defined from $\ov{b}_i$ with normalization.
Clearly, $\sum_{i=1}^m b_i=1$, and for all $1\leq i \leq m$, we have $b_i>0$ and $b_i$ can 
be expressed as a multiple of $\frac{1}{k}$, for some $k \geq \frac{m}{\eta \cdot c}$.
We have the following inequalities: for all $1 \leq i \leq m$, we have 
\[
b_i \leq a_i + \frac{1}{\ell}; \qquad b_i \geq \frac{a_i}{1 + \frac{m}{\ell}}.
\] 
The first inequality follows since $\ov{b}_i \leq a_i + \frac{1}{\ell}$ and 
$\sum_{i=1}^m \ov{b}_i \geq \sum_{i=1}^m a_i=1$.
The second inequality follows since $\ov{b}_i \geq a_i$ and 
$\sum_{i=1}^m \ov{b}_i \leq \sum_{i=1}^m (a_i + \frac{1}{\ell}) = 
\sum_{i=1}^m a_i + \frac{m}{\ell} = 1 + \frac{m}{\ell}$.
Hence for all $1 \leq i \leq m$, we have
\[
\frac{b_i}{a_i} \leq 1 + \frac{1}{\ell \cdot a_i} \leq 1 + \frac{1}{\ell \cdot c} \leq 1 + \eta;
\]
\[
\frac{a_i}{b_i} \leq 1 + \frac{m}{\ell} \leq 1 +\eta\cdot c\leq 1 + \eta.
\]
The desired result follows.
\qed
\end{proof}

\begin{lemma}{}\label{lemm:kuniform}
For all concurrent game structures $G$, for all safety objectives
$\Safe(F)$, for $F \subseteq S$,
for all $\vare>0$, there exist $k>0$ and $k$-uniform selectors $\xi$ such that
$\overline{\xi}$ is an $\vare$-optimal strategy. 
\end{lemma}
\begin{proof}
Our proof uses a result of Solan \cite{Sol03} and 
the existence of memoryless optimal strategies for concurrent 
safety games (Theorem~\ref{thrm:memory-safe}). 
We first present the result of Solan specialized for 
MDPs with reachability objectives.

\smallskip\noindent{\em The result of~\cite{Sol03}.} 
Let $G=(S,\moves,\mov_2,\trans)$ and $G'=(S,\moves,\mov_2,\trans')$ be two 
player-2 MDPs defined on the same state space $S$, with the same move set $\moves$ 
and the same move assignment function $\mov_2$, but with two different
transition functions $\trans$ and $\trans'$, respectively.
Let 
\[
\rho(G,G')= \max_{s,t \in S, a_2 \in \mov_2(s)} 
\setb{ 
\frac{\trans(s,a_2)(t)}{\trans'(s,a_2)(t)},
\frac{\trans'(s,a_2)(t)}{\trans(s,a_2)(t)}} -1;
\]
where by convention $x/0=+\infty$ for $x>0$, and $0/0=1$ 
(compare with equation (9) of~\cite{Sol03}: $\rho(G,G')$ is
obtained as a specialization of (9) of~\cite{Sol03} for MDPs).
Let $T \subseteq S$. 
For $s \in S$, let $v(s)$ and $v'(s)$ denote the 
value for player~2 for the reachability objective $\Reach(T)$ from $s$ in 
$G$ and $G'$, respectively.
Then from Theorem~6 of~\cite{Sol03} (also see equation (10) of~\cite{Sol03}) 
it follows that 
\begin{eqnarray}\label{eq-sol-sp}
-4 \cdot |S| \cdot \rho(G,G') \leq v(s) - v'(s) \leq 
\frac{4 \cdot |S| \cdot \rho(G,G') }{(1- 2\cdot|S| \cdot \rho(G,G'))^+};
\end{eqnarray}
where $x^+=\max\set{x,0}$.
We first explain how specialization of Theorem~6 of~\cite{Sol03} yields
(\ref{eq-sol-sp}).
Theorem~6 of~\cite{Sol03} was proved for value functions of discounted games 
with costs, even when the discount factor $\lambda=0$. 
Since the value functions of limit-average games are obtained as the limit of
the value functions of discounted games as the discount factor goes to $0$~\cite{MN81},
the result of Theorem~6 of~\cite{Sol03} also holds 
for concurrent limit-average games (this was the main result of~\cite{Sol03}).
Since reachability objectives are special case of limit-average objectives, Theorem~6 
of~\cite{Sol03} also holds for reachability objectives.
In the special case of reachability objectives with the same target set,
the different cost functions used in equation (10) of~\cite{Sol03} coincide, 
and the maximum absolute value of the cost is~1. Thus we obtain (\ref{eq-sol-sp})
as a specialization of Theorem~6 of~\cite{Sol03}.

We now use the existence of memoryless optimal strategies in concurrent safety 
games, and (\ref{eq-sol-sp}) to obtain our desired result.
Consider a concurrent safety game $G=(S,\moves,\mov_1,\mov_2,\trans)$ with safe set $F$ 
for player~1. 
Let $\stra_1$ be a memoryless optimal strategy for the objective $\Safe(F)$. 
Let 
$c= \min_{s \in S, a_1 \in \mov_1(s)} \set{ \stra_1(s)(a_1) \mid \stra_1(s)(a_1)>0}$
be the minimum positive transition probability given by $\stra_1$.
Given $\vare>0$, let $\eta= \min\set{\frac{1}{4\cdot|S|}, \frac{\vare}{8\cdot |S|}}$.
We define a  memoryless strategy $\stra_1'$ satisfying the following conditions: 
for $s \in S$ and $a_1 \in \mov_1(s)$ we have
\begin{enumerate}
\item if $\stra_1(s)(a_1)=0$, then $\stra_1'(s)(a_1)=0$;
\item if $\stra_1(s)(a_1)>0$, then following conditions are satisfied:
	\begin{enumerate}
	\item $\stra_1'(s)(a_1)>0$; 
	\item $\frac{\stra_1(s)(a_1)}{\stra_1'(s)(a_1)} \leq 1 + \eta$;
	\item $\frac{\stra_1'(s)(a_1)}{\stra_1(s)(a_1)} \leq 1 + \eta$;
	and 
	\item $\stra_1'(s)(a_1)$ is a multiple of $\frac{1}{k}$, for 
	an integer $k>0$ (such a $k$ exists for $k > \frac{|\moves|}{c\cdot \eta})$.
	\end{enumerate}
\end{enumerate}
For $k> \frac{|\moves|}{c \cdot \eta}$, such a strategy $\stra_1'$ exists 
(follows from the construction of Lemma~\ref{lemm-expl-constr}).
Let $G_1$ and $G_1'$ be the two player-2 MDPs obtained from $G$ by fixing the 
memoryless strategies $\stra_1$ and $\stra_1'$, respectively.
Then by definition of $\stra_1'$ we have $\rho(G_1,G_1') \leq \eta$.
Let $T=S \setminus F$. For $s \in S$, let the value of player~2 for the objective $\Reach(T)$ 
in $G_1$ and $G_1'$ be $v(s)$ and $v'(s)$, respectively.
By (\ref{eq-sol-sp}) we have 
\[
-4 \cdot |S| \cdot \eta \leq v(s) -v'(s) \leq \frac{4 \cdot |S| \cdot \eta }{(1- 2\cdot|S| \cdot \eta)^+};
\] 
Observe that by choice of $\eta$ we have 
(a)~$4\cdot |S| \cdot \eta \leq \frac{\vare}{2\cdot |S|}$ and 
(b)~$2\cdot |S| \cdot \eta \leq \frac{1}{2}$.
Hence we have $-\vare \leq v(s) -v'(s) \leq \vare$.
Since $\stra_1$ is a memoryless optimal strategy, it follows that $\stra_1'$ is a 
$k$-uniform memoryless $\vare$-optimal strategy.
\qed
\end{proof}

\smallskip\noindent{\bf Turn-based stochastic games convergence.}
We first observe that since pure memoryless optimal strategies exist for
turn-based stochastic games with safety objectives 
(the results follows from results of~\cite{Con92,LigLip69}), for turn-based stochastic games 
it suffices to iterate over pure memoryless selectors.
Since the number of pure memoryless strategies is
finite, it follows for turn-based stochastic games 
Algorithm~\ref{algorithm:strategy-improve-safe} always 
terminates and yields an optimal strategy.
In other words, we can restrict the selectors used in 
Algorithm~\ref{algorithm:strategy-improve-safe} 
in Steps 3.2.2 and 3.3.2.2 to be pure memoryless selectors.
Then the local improvement steps of Algorithm~\ref{algorithm:strategy-improve-safe}
with pure memoryless selectors terminates, and by Theorem~\ref{thrm:safe-termination}
yield a globally optimal pure memoryless strategy (which is an optimal 
strategy).
We will use the argument for turn-based stochastic games to 
a variant of Algorithm~\ref{algorithm:strategy-improve-safe} 
restricted to $k$-uniform selectors.

\medskip\noindent{\bf Strategy improvement with $k$-uniform selectors.}
We now present the variant of Algorithm~\ref{algorithm:strategy-improve-safe}
where we restrict the algorithm to $k$-uniform selectors.
The notations are essentially the same as used in 
Algorithm~\ref{algorithm:strategy-improve-safe}, but restricted to $k$-uniform 
selectors and presented as Algorithm~\ref{algo:k-uniform}.
(for example, $\ov{G}_{v_i}^k$ is similar to $\ov{G}_{v_i}$ but restricted
to $k$-uniform selectors, and similarly $\OptSel(v_i,s,k)$ are the optimal
$k$-uniform selectors, see Section~\ref{sec:appendix} for complete details).
We first argue that if we restrict 
Algorithm~\ref{algorithm:strategy-improve-safe} 
such that every iteration yields a $k$-uniform selector, for 
$k>0$, then the algorithm terminates, i.e., 
Algorithm~\ref{algo:k-uniform} terminates.
The basic argument that if Algorithm~\ref{algorithm:strategy-improve-safe} 
is restricted to $k$-uniform selectors for player~1, for $k>0$, 
then the algorithm terminates,  follows from the facts that 
(i)~the sequence of strategies obtained are monotonic (Theorem~\ref{thrm:safe-mono}) 
(i.e., the algorithm does not cycle among $k$-uniform selectors); and
(ii)~the number of $k$-uniform selectors for a given $k$ is finite.
Given $k>0$, let us denote by $z_i^k$ the valuation of 
Algorithm~\ref{algo:k-uniform} at iteration $i$.

\begin{lemma}{}\label{lemm-conv-terminate}
For all $k >0$, there exists $i \geq 0$ such that 
$z_i^k= z_{i+1}^k$.
\end{lemma}

\smallskip\noindent{\bf Convergence to optimal $k$-uniform 
memoryless strategies.}
We now argue that the valuation Algorithm~\ref{algo:k-uniform} 
converges to is optimal for $k$-uniform selectors.
The argument is as follows: if we restrict player~1 to chose
between $k$-uniform selectors, then a concurrent game structures $G$ can be 
converted to a turn-based stochastic game structure,
where player~1 first chooses a $k$-uniform selector, then player~2
chooses an action, and then the transition is determined by the
chosen $k$-uniform selector of player~1, the action of player~2
and the transition function $\trans$ of the game structure $G$.
Then by termination of turn-based stochastic games it follows
that the algorithm will terminate.
It follows from Theorem~\ref{thrm:safe-termination} that upon termination 
we obtain optimal strategy for the turn-based stochastic game.
In other words, as discussed above for turn-based stochastic game, 
the local iteration converges to a globally optimal strategy.
Hence the valuation obtained upon termination is the maximal value 
obtained over all $k$-uniform memoryless strategies. 
This gives us the following lemma (also see appendix for a detailed
proof).

\begin{lemma}{}\label{lemm-conv-optimal}
For all $k>0$, let $i\geq 0$ be such that 
$z_i^k= z_{i+1}^k$.
Then we have 
$z_i^k= \max_{\stra_1 \in \Stra_1^{M,k}} \inf_{\stra_2 \in \Stra_2} 
\Prb^{\stra_1,\stra_2}(\Safe(F))$.
\end{lemma}

\begin{lemma}{}\label{lemm-conv-bound}
For all concurrent game structures $G$, for all safety objectives
$\Safe(F)$, for $F \subseteq S$,
for all $\vare>0$, there exist $k >0$ and $i \geq 0$ such that for all 
$s \in S$ we have 
$z_i^k(s) \geq \va(\Safe(F))(s) - \vare$. 
\end{lemma}
\begin{proof}
By Lemma~\ref{lemm:kuniform}, for all $\vare>0$, there exists $k>0$
such that there is a $k$-uniform memoryless $\vare$-optimal strategy for player~1.
By Lemma~\ref{lemm-conv-terminate}, for all $k>0$, there exists an $i \geq 0$ such that 
$z_i^k=z_{i+1}^k$, and by Lemma~\ref{lemm-conv-optimal} it follows that the 
valuation $z^i_k$ represents the maximal value obtained by 
$k$-uniform memoryless strategies.
Hence it follows that there exists $k>0$ and $i \geq 0$ such that for all 
$s \in S$ we have 
$z_i^k(s) \geq \va(\Safe(F))(s) - \vare$. 
The desired result follows.
\qed
\end{proof}

We now present the convergent strategy improvement algorithm for 
safety objectives as Algorithm~\ref{algorithm:strategy-convergent}
that iterates over $k$-uniform strategy values.
The algorithm iteratively calls Algorithm~\ref{algo:k-uniform} with 
larger $k$, unless the termination condition of 
Algorithm~\ref{algorithm:strategy-improve-safe} is satisfied.

\begin{algorithm*}[t]
\caption{Convergent Safety Strategy-Improvement Algorithm}
\label{algorithm:strategy-convergent}
{
\begin{tabbing}
aaa \= aaa \= aaa \= aaa \= aaa \= aaa \= aaa \= aaa \kill
\\
\> {\bf Input:} a concurrent game structure $G$ with safe set $F$. \\
\>   {\bf Output:} a strategy $\overline{\gamma}$ for player~1. \\ 

\> 0. $k=|\moves|$ and $i=0$. \\
\> 1. {\bf do \{ } \\ 
\>\> 1.1 $\gamma_{i+1}=$ Algorithm~\ref{algo:k-uniform}($G,F,k$)\\
\>\> 1.2 Compute $v_{i+1} =\vas{\overline{\gamma}_{i+1}}(\Safe(F))$ \\
\>\> 1.3 Let $I= \set{s \in S \setminus (W_1 \cup T) \mid \Pre_1(v_i)(s) > v_i(s)}$. \\
\>\> 1.4 Let$(\ov{G}_{v_i},\ov{F})=\TB(G,v_i,F)$ \\
\>\>\> 1.4.1 let $\ov{A}_i$ be the set of almost-sure winning states in $\ov{G}_{v_i}$
	for $\Safe(\ov{F})$. \\
\>\> 1.5 Let $i=i+1$ and $k=k+1$. \\
\> {\bf \} until } $I=\emptyset$ and $(\ov{A}_{i-1} \cap S) \setminus W_1=\emptyset$. \\
\> 2. {\bf return} $\overline{\gamma}_{i}$.  
\end{tabbing}
}
\end{algorithm*}

\begin{theorem}{(Monotonicity, Optimality on termination and Convergence).}
Let $v_i$ be the valuation obtained at iteration $i$ 
of Algorithm~\ref{algorithm:strategy-convergent}.
Then the following assertions hold.
\begin{enumerate}
\item For all $i \geq 0$ we have $v_{i+1} \geq v_i$.

\item If the algorithm terminates, then $v_i=\va(\Safe(F))$.

\item For all $\vare>0$, there exists $i$ such 
that for all $s$ we have 
$v_i(s) \geq \va(\Safe(F))(s) -\vare$.

\item $\lim_{i \to \infty} v_i =\va(\Safe(F))$.

\end{enumerate}
\end{theorem}
\begin{proof} We prove the results as follows.
\begin{enumerate}
\item Let $v_i$ is the valuation of 
Algorithm~\ref{algorithm:strategy-convergent} at iteration $i$. 
For $k>0$, we consider $z^k_i$ to denote the valuation of  
Algorithm~\ref{algo:k-uniform} with the restriction of $k$-uniform selector at iteration $i$, 
and let $z^k_{i^*(k)}$ denote the least fixpoint 
(i.e., $i^*(k)$ is the least value of $i$ such that $z^k_i=z^k_{i+1}$).
Since $k$-uniform selectors are a subset of $k+1$-uniform selectors, 
it follows that the maximal value obtained over strategies that uses
$k+1$-uniform selectors is at least the maximal value obtained over $k$-uniform
selectors.
Since $z^k_{i^*(k)}$ denote the maximal value obtained over $k$-uniform selectors
(follows from Lemma~\ref{lemm-conv-optimal}), we have that 
$z^k_{i^*(k)} \leq  z^{k+1}_{i^*(k+1)}$ (note that we do not require that 
$i^*(k) \leq i^*(k+1)$, i.e., the algorithm with $k+1$-uniform selectors may require
more iterations to terminate).
We have $v_k=z^k_{i^*(k)}$ and hence the first result follows.

\item The result follows from Theorem~\ref{thrm:safe-termination}.

\item From Lemma~\ref{lemm-conv-bound} it follows that for all $\vare>0$,
there exists a $k>0$ such that for all $s$ we have 
$z^k_{i^*(k)}(s) \geq \va(\Safe(F))(s) -\vare$.
Hence $v_k \geq \va(\Safe(F))(s) -\vare$.
Hence we have that for all $\vare>0$, there exists $k \geq 0$, such that 
for all $s \in S$ we have $v_k(s) \geq \va(\Safe(F))(s) -\vare$.

\item By part (1) for all $i \geq 0$ we 
have $v_{i+1} \geq v_i$. By part (3), for all $\vare>0$, there exists $i\geq 0$
such that for all $s \in S$ we have $v_i(s) \geq \va(\Safe(F))(s) -\vare$.
Hence it follows that for all $\vare>0$, there exists $i\geq 0$ such that 
for all $j \geq i$ and for all $s \in S$ we have 
$v_j(s) \geq \va(\Safe(F))(s) -\vare$.
It follows that $\lim_{i \to \infty} v_i =\va(\Safe(F))$.
\end{enumerate}
This gives us the following result.
\qed
\end{proof}


\smallskip\noindent{\bf Discussion on convergence of Algorithm~\ref{algorithm:strategy-improve-safe}.}
We will now present an example to illustrate that (contrary to the claim of Theorem~4.3 of~\cite{CdAH09})
Algorithm~\ref{algorithm:strategy-improve-safe} need not converge to the values in 
concurrent safety games. 
However, as discussed before Algorithm~\ref{algorithm:strategy-improve-safe} satisfies the monotonicity 
and optimality on termination, and for turn-based stochastic games (and also when restricted to 
$k$-uniform strategies) converges to the values as termination is guaranteed.
In the example we will also argue how Algorithm~\ref{algorithm:strategy-convergent} 
converges to the values of the game.

\begin{example}{}\label{ex:counter-soda}
Our example consists of two steps. 
In the first step we will present a gadget where the value is irrational 
and with probability~1 absorbing states are reached.
\begin{figure}[!tb]
\begin{center}
\begin{picture}(48,28)(-5,-5)
\node[Nmarks=i, iangle=180](n0)(10,12){$s_0$}
\node[Nmarks=n](n1)(40,0){$s_1$}
\node[Nmarks=n](n2)(40,24){$s_2$}

\drawloop[loopangle=0, loopdiam=6](n1){}
\drawloop[loopangle=0, loopdiam=6](n2){}
\drawloop[ELside=l,loopCW=y, loopdiam=6](n0){$ad,1/2$}
\drawedge[ELpos=50, ELside=r, ELdist=0.5, curvedepth=6](n0,n1){$ac,bd$}
\drawedge[ELpos=50, ELside=r, ELdist=0.5, curvedepth=-6](n0,n1){$ad,1/2$}
\drawedge[ELpos=50, ELside=l, curvedepth=6](n0,n2){$bd$}
\end{picture}
\caption{A simple game with irrational value.}\label{fig:ex1}
\end{center}
\end{figure}
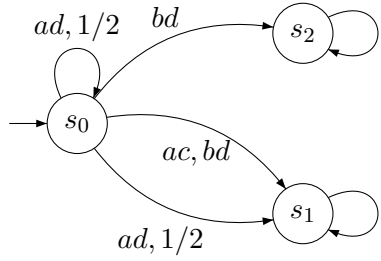

\smallskip\noindent{\bf Step 1.} We first consider the game shown in
Fig~\ref{fig:ex1} with three states $\set{s_0,s_1,s_2}$ with two actions
$a,b$ for player~1 and $c,d$ for player~2. 
The states $s_0,s_1$ are safe states, and $s_2$ is a non-safe state. 
The transitions are as follows: (1)~$s_1$ and $s_2$ are absorbing; 
and (2)~in $s_0$ we have the following transitions, 
(a)~given action pairs $ac$ and $bd$ the next state is $s_1$,
(b)~given action pair $bc$ the next state is $s_2$, and 
(c)~given action pair $ad$ the next states are $s_0$ and $s_1$ with 
probability $1/2$ each.
In this game, the state $s_0$ is transient, as given any action pairs,
the set  $\set{s_1,s_2}$ of absorbing states is reached with probability at least $1/2$ in 
one step.
Hence the set $\set{s_1,s_2}$ is reached with probability~1, irrespective
of the choice of strategies of the players.
Hence in this game the objective for player~1 is equivalently to reach 
$s_1$.
Let us denote by $x$ the value of the game at $s_0$, and let 
us consider the following matrix 
\[
M=\begin{bmatrix}
    \ 1 \ & \ \frac{x}{2} \  \\
    \ 0 \ & \ 1 \  
  \end{bmatrix}
\]
Then $x=\min \max M$.
In other words, consider the valuation $v_x=(x,1,0)$ for states $s_0,s_1$ and 
$s_2$, respectively, and $x=\min\max M$ describes that $v_x=\Pre_1(v_x)$, and 
it is the least fixpoint of valuations satisfying $v=\Pre_1(v)$.
We now analyze the value $x$ at $s_0$.
The solution of $x$ is achieved by solving the following optimization 
problem
\[
\text{minimize $x$} \quad \text{ subject to } \quad 
y + \big((1-y)\cdot x\big)/2 \leq x \text{ and }
1-y \leq x.
\]
Intuitively, $y$ denotes the probability to choose move $a$ in an optimal 
strategy. The solution to the optimization problem is achieved by setting 
$x=1-y$.
Hence we have $y + (1-y)^2/2 =(1-y)$, i.e., $(1+y)^2=2$. 
Since $y$ must lie in the interval $[0,1]$, we have $y=\sqrt{2}-1$, and 
thus we have $x=2-\sqrt{2} < 0.6$.
We now analyze Algorithm~\ref{algorithm:strategy-improve-safe} on this example.
Let $v_i$ denote the valuation of the $i$-th iteration, and let $v_i^0$ be the value at
state $s_0$.
Then we have $v_i^0 < v_{i+1}^0$ and in the limit it converges to 
value $2-\sqrt{2}$.
We observe that on this example Algorithm~\ref{algorithm:strategy-improve-safe}
exactly behaves as Algorithm~\ref{algorithm:strategy-improve} (strategy improvement
for reachability) as the objective for player~1 is equivalently
to reach $s_1$, since $s_0$ is transient.
The reason of the strict inequality $v_i^0 < v_{i+1}^0$ is as follows: 
if the valuation at state $s_0$ in $i$-th and $i+1$-th iteration is the same,
then by correctness of Algorithm~\ref{algorithm:strategy-improve} it follows 
that the values would have been achieved in finitely many steps, implying 
convergence to a rational value at $s_0$. 
The convergence to the values in the limit is due to correctness of 
Algorithm~\ref{algorithm:strategy-improve}.

\begin{figure}[!tb]
\begin{center}
\begin{picture}(48,28)(-5,-5)
\node[Nmarks=n](n0)(40,12){$s_0$}
\node[Nmarks=n](n1)(70,0){$s_1$}
\node[Nmarks=n](n2)(70,24){$s_2$}

\node[Nmarks=n](n3)(20,12){$s_3$}
\node[Nmarks=n](n4)(0,12){$s_4$}
\node[Nmarks=n](n5)(-20,12){$s_5$}

\drawloop[loopangle=0, loopdiam=6](n1){}
\drawloop[loopangle=0, loopdiam=6](n2){}
\drawloop[ELside=l,loopCW=y, loopdiam=6](n0){$ad,1/2$}
\drawedge[ELpos=50, ELside=r, ELdist=0.5, curvedepth=6](n0,n1){$ac,bd$}
\drawedge[ELpos=50, ELside=r, ELdist=0.5, curvedepth=-6](n0,n1){$ad,1/2$}
\drawedge[ELpos=50, ELside=l, curvedepth=6](n0,n2){$bd$}
\drawedge[ELpos=50, ELside=l, curvedepth=0](n3,n0){$a\bot$}
\drawedge[ELpos=50, ELside=l, curvedepth=0](n4,n5){$\bot d$}
\drawedge[ELpos=50, ELside=l, curvedepth=6](n3,n4){$b\bot$}
\drawedge[ELpos=50, ELside=l, curvedepth=6](n4,n3){$\bot c$}
\drawedge[ELpos=50, ELside=l, curvedepth=0](n4,n5){$\bot d$}
\drawedge[ELpos=50, ELside=l, curvedepth=0](n4,n5){$\bot d$}
\end{picture}
\caption{Counter example game.}\label{fig:ex2}
\end{center}
\end{figure}
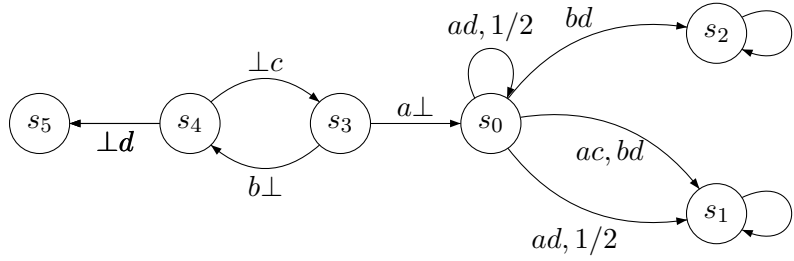

\smallskip\noindent{\bf Step~2.} We will now augment the game of 
Step~1 to construct an example to show that Algorithm~2 does not necessarily
converge to the values.
Consider the game shown in Fig~\ref{fig:ex2} augmenting the game of 
Fig~\ref{fig:ex1} with some additional states (states
$s_3,s_4$ and $s_5$) and transitions (we only show the interesting transitions 
in the figure for simplicity).
All the additional states shown are safe states. 
The value of state $s_5$ is $0.6$ (consider it as a probabilistic state 
going to state $s_1$ with probability $0.6$ and $s_2$ with probability 
$0.4$, and these edges are not shown in the figure).
The transitions at state $s_3$ and $s_4$ are as follows:
in state $s_3$, player~1 can goto state $s_0$ or $s_4$ by choosing 
actions $a$ and $b$, respectively (at $s_3$ player~2 has only one action $\bot$);
and in state $s_4$, player~2 can goto state $s_3$ or $s_5$ by choosing 
actions $c$ and $d$, respectively (at $s_4$ player~1 has only one action $\bot$).
We analyze Algorithm~\ref{algorithm:strategy-improve-safe} on the example
shown in Fig~\ref{fig:ex2}.
In this game, at $s_3$ player~1 starts by playing actions $a$ and $b$ 
uniformly, and player~2 responds by chosing action $c$. 
In the iterations of the algorithm it follows by the argument of 
Step~1, that the set $I$ of Step~3.1 of Algorithm~\ref{algorithm:strategy-improve-safe} 
is always non-empty as $s_0 \in I$. 
Hence in every iteration the value at $s_0$ improves, and the strategy in 
$s_3$ and $s_4$ does not change.
Hence the valuation at $s_3$ converges to the valuation at $s_0$, i.e., 
to $2-\sqrt{2} < 0.6$.
However, by switching to action $b$ at $s_3$, player~1 can enforce player~2 to
play action $d$ at $s_4$ and ensure value $0.6$.
In other words,  the value at $s_3$ is $0.6$, whereas Algorithm~\ref{algorithm:strategy-improve-safe}
converges to $2-\sqrt{2}< 0.6$.

The switching to action $b$ would have been ensured by the turn-based construction 
of Step~3.3. 
For turn-based stochastic games or $k$-uniform memoryless strategies, since
convergence to values is guaranteed, the turn-based construction of Step~3.3 
is also ensured to get executed.
However, as the convergence to values in concurrent games is in the limit,
Step~3.3 of Algorithm~\ref{algorithm:strategy-improve-safe} may not get executed 
as shown by this example. 
However, we now illustrate that the valuations of 
Algorithm~\ref{algorithm:strategy-convergent} converges to the values.
We consider Algorithm~\ref{algorithm:strategy-convergent}: 
Consider $k$-uniform strategies, for a finite $k\geq 2$, then the value at 
$s_0$ for $k$-uniform strategies converges in 
finitely many steps to a value smaller than $0.6$ (as it converges to a value smaller than the value at $s_0$), 
and then Step~3.3 of Algorithm~\ref{algo:k-uniform} would get executed, and the 
value at $s_3$ would be assigned to~$0.6$.
In other words, for Algorithm~\ref{algorithm:strategy-convergent} the values at 
$s_3,s_4$ and $s_5$ are always set to $0.6$, and the 
value at $s_0$ converges in the limit to $2-\sqrt{2}$.
Thus with the example we show that though Algorithm~\ref{algorithm:strategy-improve-safe}
does not necessarily converge to the values, Algorithm~\ref{algorithm:strategy-convergent}
correctly converges to the values.
\qed
\end{example}



\smallskip\noindent{\bf Retraction of Theorem~4.3 of~\cite{CdAH09}.} 
In~\cite{CdAH09}, the convergence of Algorithm~\ref{algorithm:strategy-improve-safe} 
was claimed.
Unfortunately the theorem is incorrect (with irreparable error) as shown by 
Example~\ref{ex:counter-soda} and we retract the claim of Theorem~4.3 of~\cite{CdAH09} 
of convergence of Algorithm~\ref{algorithm:strategy-improve-safe} for concurrent games.

\smallskip\noindent{\bf Complexity.} 
Algorithm~\ref{algorithm:strategy-improve-safe} may not terminate
in general; we briefly describe the complexity of every iteration.
Given a valuation $v_i$, the computation of $\Pre_1(v_i)$ 
involves the solution of matrix games with rewards $v_i$; this can be
done in polynomial time using linear programming.
Given $v_i$, if $\Pre_1(v_i)=v_i$, 
the sets $\OptSel(v_i,s)$ and $\OptSelCount(v_i,s)$ 
can be computed by enumerating the subsets of available actions
at $s$ and then using linear-programming.
For example, to check whether
$(A,B)\in \OptSelCount(v_i,s)$ it suffices to check both of these facts:
\begin{enumerate}
\item \emph{($A$ is the support of an optimal selector $\xi_1$).} 
there is an selector $\xi_1$ such that 
(i)~$\xi_1$ is optimal (i.e. for all actions $b \in \mov_2(s)$ we have 
$\Pre_{\xi_1,b}(v_i)(s)\geq v_i(s)$);
(ii)~for all $a \in A$ we have $\xi_1(a)>0$, and for all
$a \not \in A$ we have $\xi_1(a)=0$;

\item \emph{($B$ is the set of counter-optimal actions against $\xi_1$).}
for all $b \in B$ we have $\Pre_{\xi_1,b}(v_i)(s)= v_i(s)$, 
and for all $b \not\in B$ we have $\Pre_{\xi_1,b}(v_i)(s)> v_i(s)$.

\end{enumerate}
All the above checks can be performed by checking feasibility of
sets of linear equalities and inequalities.
Hence, $\TB(G,v_i,F)$ can be computed in time 
polynomial in size of $G$ and $v_i$ and exponential in the 
number of moves.
We observe that the construction is exponential only in the number of
moves at a state, and not in the number of states. 
The number of moves at a state is typically much smaller than the size
of the state space.
We also observe that the improvement step 3.3.2 requires the
computation of the set of almost-sure winning states of a turn-based
stochastic safety game: this can be done both via linear-time discrete
graph-theoretic algorithms \cite{CJH03}, and via symbolic
algorithms~\cite{crg-tcs07}. 
Both of these methods are more efficient than the basic step 3.4 of
the improvement algorithm, where the quantitative values of an MDP
must be computed. 
Thus, the improvement step 3.3 of
Algorithm~\ref{algorithm:strategy-improve-safe} is in practice should not
be inefficient, compared with the standard improvement steps 3.2 and 3.4.
We now discuss the above steps for Algorithm~\ref{algo:k-uniform}.  
The argument is similar as above, but in case of $k$-uniform selectors, 
we need to ensure that the witness selectors are $k$-uniform which can 
be achieved with integer constraints. 
In other words, for Algorithm~\ref{algo:k-uniform} the above checks are
performed by checking feasibility of sets of integer linear equalities
and inequalities (which can be achieved in exponential time).
Again, the construction is exponential in the number of moves at a state,
and not in the number of states.
Hence we enumerate over sets of moves at a state (exponential in number of
moves), and then need to solve integer linear constraints (the size of 
the integer linear constraints is polynomial in the number of moves, 
and is achieved in time exponential in the number of moves).
Thus again the improvement step 3.3 of Algorithm~\ref{algo:k-uniform}
is polynomial in the size of the game, and exponential in the number of
moves.

\subsection{Termination for Approximation }
In this subsection we present termination criteria for strategy improvement
algorithms for concurrent games for $\vare$-approximation.

\medskip\noindent{\bf Termination for concurrent games.}
We apply the reachability strategy improvement algorithm 
(Algorithm~\ref{algorithm:strategy-improve}) for player~2, 
for a reachability objective $\Reach(T)$, we obtain a
sequence of valuations $(u_i)_{i\geq 0}$ such that 
(a) $u_{i+1} \geq u_i$;
(b) if $u_{i+1}=u_i$, then $u_i=\vb(\Reach(T))$; and
(c) $\lim_{i \to \infty} u_i =\vb(\Reach(T))$.
Given a concurrent game $G$ with $F \subs S$ and $T=S\setminus F$,
we apply Algorithm~\ref{algorithm:strategy-improve} to obtain
the sequence of valuation $(u_i)_{i\geq 0}$ as above, and 
we apply Algorithm~\ref{algorithm:strategy-convergent} 
to obtain a sequence of valuation $(v_i)_{i \geq 0}$.
The termination criteria are as follows:
\begin{enumerate}
\item if for some $i$ we have $u_{i+1}=u_i$, then we have 
$u_i=\vb(\Reach(T))$, and $1-u_i=\va(\Safe(F))$, and we 
obtain the values of the game;
\item if for some $i$ we have $v_{i+1}=v_i$, then we have 
$1-v_i=\vb(\Reach(T))$, and $v_i=\va(\Safe(F))$, and we 
obtain the values of the game; and
\item for $\vare>0$, if for some $i\geq 0$, we have $u_i + v_i \geq 1-\vare$,
then for all $s \in S$ we have
$v_i(s)\geq \va(\Safe(F))(s) -\vare$ and
$u_i(s)\geq \vb(\Reach(T))(s) -\vare$ (i.e., the algorithm
can stop for $\vare$-approximation).
\end{enumerate}
Observe that since $(u_i)_{i\geq 0}$ and $(v_i)_{i \geq 0}$ are 
both monotonically non-decreasing and $\va(\Safe(F))+ \vb(\Reach(T))=1$, 
it follows that if $u_i + v_i \geq 1-\vare$, then forall 
$j\geq i$ we have $u_i \geq u_j -\vare$ and
$v_i \geq v_j -\vare$.
This establishes that $u_i \geq \va(\Safe(F)) -\vare$ and
$v_i \geq \vb(\Reach(T)) -\vare$;
and the correctness of the stopping criteria (3) for 
$\vare$-approximation follows.
We also note that instead of applying the reachability 
strategy improvement algorithm, a value-iteration algorithm
can be applied for reachability games to obtain a 
sequence of valuation with properties similar to $(u_i)_{i \geq 0}$
and the above termination criteria can be applied.

\begin{theorem}{}
Let $G$ be a concurrent game structure with a safety objective $\Safe(F)$.
Algorithm~\ref{algorithm:strategy-convergent} and 
Algorithm~\ref{algorithm:strategy-improve} for player~2 for the 
reachability objective $\Reach(S\setminus F)$ yield two sequences of monotonic 
valuations $(v_i)_{i \geq 0}$ and $(u_i)_{i \geq 0}$, respectively, such that 
(a)~for all $i\geq 0$, we have $v_i \leq \va(\Safe(F)) \leq 1-u_i$; 
and
(b)~$\lim_{i \to \infty} v_i = \lim_{i \to \infty} 1-u_i = \va(\Safe(F))$. 
\end{theorem}

\smallskip\noindent{\bf Bounds for approximation.} We now discuss the 
bounds for approximation for concurrent games with reachability objectives, 
which follows from the results of~\cite{HKM09,CSR11}.
It follows from the results of~\cite{HKM09} that for all $\vare>0$, 
there exist $k$-uniform memoryless optimal strategies for concurrent 
reachability and safety games $G$, where $k$ is bounded by 
$(\frac{1}{\vare})^{2^{O(|G|)}}$.
It follows that for all $\vare>0$, if we consider our strategy improvement 
algorithm (restricted to $k$-uniform selectors) for reachability games,
then upon termination the valuation obtained is an $\vare$-approximation 
of the value function of the game, where $k$ is bounded by 
$(\frac{1}{\vare})^{2^{O(|G|)}}$.
Using the restriction to $k$-uniform memoryless strategies, along with 
the reduction of concurrent games to turn-based stochastic game for 
$k$-uniform memoryless strategies and the termination bound for turn-based
stochastic games we obtain a double exponential bound on the number of 
iterations required for termination (note that if $k=(\frac{1}{\vare})^{2^{O(|G|)}}$,
then the total number of $k$-uniform memoryless strategies is $k^{O(|G|)}$,
which is double exponential) (also see~\cite{CSR11} for details).
Moreover, the recent result of~\cite{CSR11} shows that the double exponential 
bound is near optimal for the strategy improvment algorithm for 
concurrent games with reachability objectives.

\smallskip\noindent{\bf Approximation of strategies.}
The previous method to solve concurrent reachability and safety games 
was the value-iteration algorithm. 
The witness strategy produced by the value-iteration algorithm for concurrent 
reachability games is not memoryless; and for concurrent safety games since
the value-iteration algorithm converges from above it does not provide any 
witness strategies. 
The only previous algorithm to approximate memoryless $\vare$-optimal 
strategies, for $\vare>0$, for concurrent reachability and safety games is 
the naive algorithm that exhaustively searches over the set of all 
$k$-uniform memoryless strategies (such that the $k$-uniform memoryless 
strategies suffices for $\vare$-optimality and $k$-depends in $\vare$).
Our strategy improvement algorithms for concurrent reachability and safety games
are the first strategy search based approach to approximate $\vare$-optimal strategies.


\medskip\noindent{\bf Acknowledgements.}
We thanks anonymous reviewers for many helpful and insightful
comments that significantly improved the presentation of the paper, and
help us fix gaps in the results of the conference versions.
We warmly acknowlede their help.
This work was partially supported by ERC Start Grant Graph Games (Project No 279307), 
FWF NFN Grant S11407-N23 (RiSE) and a Microsoft faculty fellowship.

\clearpage

\section{Technical Appendix}\label{sec:appendix}

We now present the details of restriction to $k$-uniform selectors, and 
the details of the notations used in Algorithm~\ref{algo:k-uniform}.
The definitions are essentially same as for selectors and optimal 
selectors, but restricted to $k$-uniform selectors.

\medskip\noindent{\bf Optimal $k$-uniform selectors.} 
For $k>0$, a valuation $v$ and a state $s$, let 
\[
\Pre_1^k(v)(s)= \sup_{\xi_1' \in \Sel_1^k(s)} \Pre_{1: \xi_1'}(v)(s).
\]
denote the optimal one-step value among $k$-uniform selectors.
For $k>0$, given a valuation $v$ and a state $s$, we define by
\[
\OptSel(v,s,k) =\set{\xi_1 \in \Sel_1^k(s)  \mid 
\Pre_{1:\xi_1}(v)(s) = \Pre_{1}^k(v)(s)}
\]
the set of optimal selectors among $k$-uniform selectors for $v$ at state $s$.
For a $k$-uniform optimal selector $\xi_1 \in \OptSel(v,s,k)$, we define the 
set of counter-optimal actions as follows:
\[
\CountOpt(v,s,\xi_1,k) =\set{ b \in \mov_2(s) \mid 
\Pre_{\xi_1,b}(v)(s) = \Pre_{1}^k(v)(s)}.
\]
Observe that for $\xi_1 \in \OptSel(v,s,k)$, for all $b \in 
\mov_2(s) \setminus \CountOpt(v,s,\xi_1,k)$ we have 
$\Pre_{\xi_1,b}(v)(s) >  \Pre_1^k(v)(s)$.
We define the set of $k$-uniform optimal selector support and the 
counter-optimal action set as follows:
\[
\begin{array}{rcl}
\OptSelCount(v,s,k) & = & \set{(A,B) \subs \mov_1(s) \times \mov_2(s) \mid
\exists \xi_1 \in \Sel_1^k(s). \ \xi_1 \in \OptSel(v,s,k) \\
 & & \land 
\ \ \supp(\xi_1)=A \ \land \ \CountOpt(v,s,\xi_1,k)=B
};
\end{array}
\]
i.e., it consists of pairs $(A,B)$ of actions of player~1 and player~2,
such that there is a $k$-uniform optimal selector $\xi_1$ with support $A$,
and $B$ is the set of counter-optimal actions to $\xi_1$.

\medskip\noindent{\bf Turn-based reduction.} Given a concurrent 
game $G=\langle S,\moves,\mov_1,\mov_2, \trans \rangle $,  
a valuation $v$, and bound $k$ for $k$-uniformity we construct a 
turn-based stochastic game
$\ov{G}_v^k=\langle (\ov{S},\ov{E}), (\ov{S}_1,\ov{S}_2,\ov{S}_R),\ov{\trans}
\rangle$ as follows:
\begin{enumerate}
\item The set of states is as follows:
\[
\begin{array}{rcl}
\ov{S}& = & S \cup \set{(s,A,B) \mid s\in S, \ (A,B) \in \OptSelCount(v,s,k)} \\
	&\cup & \set{(s,A,b) \mid s \in S, \ (A,B) \in \OptSelCount(v,s,k), \ b \in B}.
\end{array}
\]

\item The state space partition is as follows: 
$\ov{S}_1=S$; $\ov{S}_2=\set{(s,A,B) \mid s \in S, (A,B) \in \OptSelCount(v,s,k)}$;
and $\ov{S}_R=\set{(s,A,b) \mid s\in S ,\  (A,B) \in \OptSelCount(v,s,k), b \in B}$.
In other words, $(\ov{S}_1,\ov{S}_2,\ov{S}_R)$ is a partition of the state 
space, where $\ov{S}_1$ are player~1 states, $\ov{S}_2$ are player~2 states,
and $\ov{S}_R$ are random or probabilistic states.

\item The set of edges is as follows:
\[ 
\begin{array}{rcl}
\ov{E} & = & \set{(s,(s,A,B)) \mid s \in S, (A,B) \in \OptSelCount(v,s,k)} \\
	& \cup & \set{((s,A,B),(s,A,b)) \mid b \in B} 
	\cup \set{((s,A,b),t) \mid \displaystyle t \in \bigcup_{a \in A} \dest(s,a,b)}.
\end{array}
\]

\item The transition function $\ov{\trans}$ for all states in $\ov{S}_R$ 
is uniform over its successors.
\end{enumerate}
Intuitively, the reduction is as follows.
Given the valuation $v$, state $s$ is a player~1 state where
player~1 can select a pair $(A,B)$ (and move to
state $(s,A,B)$) with $A \subs \mov_1(s)$ 
and $B \subs \mov_2(s)$ such that there is a $k$-uniform optimal 
selector $\xi_1$ with support exactly $A$ and the set of
counter-optimal actions to $\xi_1$ is the set $B$.
From a player~2 state $(s,A,B)$, player~2 can choose any action
$b$ from the set $B$, and move to state $(s,A,b)$.
A state $(s,A,b)$ is a probabilistic state where all the 
states in $\bigcup_{a\in A} \dest(s,a,b)$ are chosen 
uniformly at random.
Given a set $F \subseteq S$ we denote by $\ov{F}= F \cup 
\set{(s,A,B) \in\ov{S} \mid s \in F} \cup 
\set{(s,A,b) \in\ov{S} \mid s \in F}$.
We refer to the above reduction as $\TB$, i.e., 
$(\ov{G}_v^k,\ov{F})=\TB(G,v,F,k)$.

\begin{proof} {\bf (of Lemma~\ref{lemm-conv-optimal}).}
The proof of the result is essentially identical as the proof of 
Theorem~\ref{thrm:safe-termination}, and we present the details
for completeness.
Let $v_i=z_i^k$.
We show that for all $k$-uniform memoryless strategies $\stra_1$ for player~1 
we have $\vas{\stra_1}(\Safe(F)) \leq v_i$.

Let $\ov{\stra}_2$ be a pure memoryless optimal strategy for 
player~2 in $\ov{G}_{v_i}^k$ for the objective 
complementary to $\Safe(\ov{F})$, 
where $(\ov{G}_{v_i}^k,\Safe(\ov{F}))=\TB(G,v_i,F,k)$.
Consider a $k$-uniform memoryless strategy $\stra_1$ for player~1,
and we define a pure memoryless strategy $\stra_2$
for player~2 as follows.
\begin{enumerate}
\item If $\stra_1(s) \not\in \OptSel(v_i,s,k)$, then $\stra_2(s)=b \in \mov_2(s)$,
	such that $\Pre_{\stra_1(s),b}(v_i)(s) < v_i(s)$;
	(such a $b$ exists since $\stra_1(s) \not\in \OptSel(v_i,s,k)$).

\item If $\stra_1(s) \in \OptSel(v_i,s,k)$, then let $A=\supp(\stra_1(s))$,
	and consider $B$ such that $B=\CountOpt(v_i,s,\stra_1(s),k)$.
	Then we have $\stra_2(s)=b$, such that $\ov{\stra}_2((s,A,B))=(s,A,b)$. 
\end{enumerate}
Observe that by construction of $\stra_2$, for all 
$s \in S \setminus (W_1 \cup T)$, we have 
$\Pre_{\stra_1(s),\stra_2(s)}(v_i)(s) \leq v_i(s)$.
We first show that in the Markov chain obtained by fixing $\stra_1$ and 
$\stra_2$ in $G$, there is no closed connected recurrent set of states $C$
such that $C \subseteq S \setminus (W_1 \cup T)$.
Assume towards contradiction that $C$ is a closed connected recurrent 
set of states in $S \setminus (W_1 \cup T)$.
The following case analysis achieves the contradiction.
\begin{enumerate}
\item Suppose for every state $s \in C$ we have $\stra_1(s) \in \OptSel(v_i,s,k)$.
Then consider the strategy $\ov{\stra}_1$ in $\ov{G}_{v_i}^k$ such that 
for a state $s \in C$ we have $\ov{\stra}_1(s)=(s,A,B)$,
where $\stra_1(s)=A$, and $B=\CountOpt(v_i,s,\stra_1(s),k)$.
Since $C$ is closed connected recurrent states, it follows by construction 
that for all states $s \in C$ in the game $\ov{G}_{v_i}^k$ we have 
$\Prb_s^{\ov{\stra}_1,\ov{\stra}_2}(\Safe(\ov{C}))=1$,
where $\ov{C}=C \cup \set{(s,A,B) \mid s \in C} \cup \set{(s,A,b) \mid s \in C}$.
It follows that for all $s \in C$ in $\ov{G}_{v_i}^k$ 
we have $\Prb_s^{\ov{\stra}_1,\ov{\stra}_2}(\Safe(\ov{F}))=1$.
Since $\ov{\stra}_2$ is an optimal strategy, it follows that $C 
\subseteq (\ov{A}_i^k \cap S)\setminus W_1$.
This contradicts that $(\ov{A}_i^k \cap S) \setminus W_1=\emptyset$.

\item Otherwise for some state $s^* \in C$ we have $\stra_1(s^*) \not\in
\OptSel(v_i,s^*,k)$.
Let $r=\min\set{q \mid U_q(v_i) \cap C \neq \emptyset}$, i.e., 
$r$ is the least value-class with non-empty intersection with $C$.
Hence it follows that for all $q<r$, we have 
$U_q(v_i) \cap C=\emptyset$.
Observe that since for all $s \in C$ we have 
$\Pre_{\stra_1(s),\stra_2(s)}(v_i)(s) \leq v_i(s)$,
it follows that for all $s \in U_r(v_i)$ either
(a)~$\dest(s,\stra_1(s),\stra_2(s))\subseteq U_r(v_i)$;
or (b)~$\dest(s,\stra_1(s),\stra_2(s)) \cap U_q(v_i) \neq \emptyset$,
for some $q<r$.
Since $U_r(v_i)$ is the least value-class with non-empty intersection 
with $C$, it follows that for all $s \in U_r(v_i)$ we have 
$\dest(s,\stra_1(s),\stra_2(s)) \subseteq U_r(v_i)$.
It follows that $C \subseteq U_r(v_i)$. 
Consider the state $s^* \in C$ such that $\stra_1(s^*) \not\in \OptSel(v_i,s,k)$.
By the construction of $\stra_2(s)$, we have 
$\Pre_{\stra_1(s^*),\stra_2(s^*)}(v_i)(s^*) < v_i(s^*)$.
Hence we must have $\dest(s^*,\stra_1(s^*),\stra_2(s^*)) \cap U_{q}(v_i) 
\neq \emptyset$, for some $q <r$.
Thus we have a contradiction.
\end{enumerate}
It follows from above that there is no closed connected recurrent set of states
in $S\setminus (W_1 \cup T)$, and hence with probability~1 
the game reaches $W_1 \cup T$ from all states in $S \setminus (W_1 \cup T)$.
Hence the probability to satisfy $\Safe(F)$ is equal to the probability 
to reach $W_1$.
Since for all states $s \in S\setminus (W_1 \cup T)$ we have 
$\Pre_{\stra_1(s),\stra_2(s)}(v_i)(s) \leq v_i(s)$, 
it follows that given the strategies $\stra_1$ and 
$\stra_2$, the valuation $v_i$ satisfies all the inequalities 
for linear program to reach $W_1$.
It follows that the probability to reach $W_1$ from $s$ is 
atmost $v_i(s)$.
It follows that for all $s \in S\setminus (W_1 \cup T)$ 
we have $\vas{\stra_1}(\Safe(F))(s)\leq v_i(s)$.
This completes the proof.
\qed
\end{proof}


\begin{thebibliography}{10}

\bibitem{Bertsekas95}
D.P. Bertsekas.
\newblock {\em Dynamic Programming and Optimal Control}.
\newblock Athena Scientific, 1995.
\newblock Volumes I and II.

\bibitem{CdAH09}
K.~Chatterjee, L.~de~Alfaro, and T.~A. Henzinger.
\newblock Termination criteria for solving concurrent safety and reachability
  games.
\newblock In {\em SODA}, pages 197--206. ACM-SIAM, 2009.

\bibitem{CdAH06}
K.~Chatterjee, L.~de~Alfaro, and T.A. Henzinger.
\newblock Strategy improvement in concurrent reachability games.
\newblock In {\em QEST'06}. IEEE, 2006.

\bibitem{CJH03}
K.~Chatterjee, M.~Jurdzi{\'n}ski, and T.A. Henzinger.
\newblock Simple stochastic parity games.
\newblock In {\em CSL'03}, volume 2803 of {\em LNCS}, pages 100--113. Springer,
  2003.

\bibitem{Con92}
A.~Condon.
\newblock The complexity of stochastic games.
\newblock {\em Information and Computation}, 96(2):203--224, 1992.

\bibitem{Con93}
A.~Condon.
\newblock On algorithms for simple stochastic games.
\newblock In {\em Advances in Computational Complexity Theory}, volume~13 of
  {\em DIMACS Series in Discrete Mathematics and Theoretical Computer Science},
  pages 51--73. American Mathematical Society, 1993.

\bibitem{CY95}
C.~Courcoubetis and M.~Yannakakis.
\newblock The complexity of probabilistic verification.
\newblock {\em Journal of the ACM}, 42(4):857--907, 1995.

\bibitem{luca-thesis}
L.~de~Alfaro.
\newblock {\em Formal Verification of Probabilistic Systems}.
\newblock PhD thesis, Stanford University, 1997.
\newblock Technical Report STAN-CS-TR-98-1601.

\bibitem{dAH00}
L.~de~Alfaro and T.A. Henzinger.
\newblock Concurrent omega-regular games.
\newblock In {\em Proceedings of the 15th Annual Symposium on Logic in Computer
  Science}, pages 141--154. IEEE Computer Society Press, 2000.

\bibitem{crg-tcs07}
L.~de~Alfaro, T.A. Henzinger, and O.~Kupferman.
\newblock Concurrent reachability games.
\newblock {\em Theoretical Computer Science}, 386(3):188--217, 2007.

\bibitem{dAM04}
L.~de~Alfaro and R.~Majumdar.
\newblock Quantitative solution of omega-regular games.
\newblock {\em Journal of Computer and System Sciences}, 68:374--397, 2004.

\bibitem{Derman}
C.~Derman.
\newblock {\em Finite State {Markovian} Decision Processes}.
\newblock Academic Press, 1970.

\bibitem{EY06}
K.~Etessami and M.~Yannakakis.
\newblock Recursive concurrent stochastic games.
\newblock In {\em ICALP 06: Automata, Languages, and Programming}. Springer,
  2006.

\bibitem{Eve57}
H.~Everett.
\newblock Recursive games.
\newblock In {\em Contributions to the Theory of Games {III}}, volume~39 of
  {\em Annals of Mathematical Studies}, pages 47--78, 1957.

\bibitem{FV97}
J.~Filar and K.~Vrieze.
\newblock {\em Competitive {Markov} Decision Processes}.
\newblock Springer-Verlag, 1997.

\bibitem{GH07}
H.~Gimbert and F.~Horn.
\newblock Simple stochastic games with few random vertices are easy to solve.
\newblock In {\em FoSSaCS'08}, 2008.

\bibitem{CSR11}
K.~A. Hansen, R.~Ibsen-Jensen, and P.~B. Miltersen.
\newblock The complexity of solving reachability games using value and strategy
  iteration.
\newblock In {\em CSR}, pages 77--90, 2011.

\bibitem{HKM09}
K.~A. Hansen, M.~Kouck{\'y}, and P.~. Miltersen.
\newblock Winning concurrent reachability games requires doubly-exponential
  patience.
\newblock In {\em LICS}, pages 332--341, 2009.

\bibitem{HofKar66}
A.J. Hoffman and R.M. Karp.
\newblock On nonterminating stochastic games.
\newblock {\em Management Sciences}, 12(5):359--370, 1966.

\bibitem{Howard}
R.~A. Howard.
\newblock {\em Dynamic Programming and {Markov} Processes}.
\newblock MIT Press, 1960.

\bibitem{Kemeny}
J.G. Kemeny, J.L. Snell, and A.W. Knapp.
\newblock {\em Denumerable {Markov} Chains}.
\newblock D. Van Nostrand Company, 1966.

\bibitem{LigLip69}
T.~A. Liggett and S.~A. Lippman.
\newblock Stochastic games with perfect information and time average payoff.
\newblock {\em Siam Review}, 11:604--607, 1969.

\bibitem{MN81}
J.F. Mertens and A.~Neyman.
\newblock Stochastic games.
\newblock {\em International Journal of Game Theory}, 10:53--66, 1981.

\bibitem{TParth71}
T.~Parthasarathy.
\newblock Discounted and positive stochastic games.
\newblock {\em Bulletin of American Mathematical Society}, 77(1):134--136,
  1971.

\bibitem{RCN73}
S.~S. Rao, R.~Chandrasekaran, and K.~P.~K. Nair.
\newblock Algorithms for discounted games.
\newblock {\em Journal of Opt. Theory and Applications}, pages 627--637, 1973.

\bibitem{Sha53}
L.S. Shapley.
\newblock Stochastic games.
\newblock {\em Proc.\ Nat.\ Acad.\ Sci. {USA}}, 39:1095--1100, 1953.

\bibitem{Sol03}
E.~Solan.
\newblock Continuity of the value of competitive {M}arkov decision processes.
\newblock {\em Journal of Theoretical Probability}, 16:831--845, 2003.

\bibitem{ZP95}
U.~Zwick and M.S. Paterson.
\newblock The complexity of mean payoff games on graphs.
\newblock {\em Theoretical Computer Science}, 158:343--359, 1996.

\end{thebibliography}
\end{document}